\documentclass{llncs}
\let\origvec\vec

\let\vec\origvec
\usepackage{amsmath}
\usepackage{amssymb}
\usepackage{framed}
\usepackage{graphicx}
\usepackage{etoolbox}

\usepackage[matrix,frame,arrow]{xy}
\usepackage{amsmath}
\newcommand{\qw}[1][-1]{\ar @{-} [0,#1]}
\newcommand{\multigate}[2]{*+<1em,.9em>{\hphantom{#2}} \qw \POS[0,0].[#1,0];p !C *{#2},p \save+LU;+RU **\dir{-}\restore\save+RU;+RD **\dir{-}\restore\save+RD;+LD **\dir{-}\restore\save+LD;+LU **\dir{-}\restore}
\newcommand{\ghost}[1]{*+<1em,.9em>{\hphantom{#1}} \qw}

\newcommand{\rstick}[1]{*!L!<-.5em,0em>=<0em>{#1}}
\newcommand{\lstick}[1]{*!R!<.5em,0em>=<0em>{#1}}

\newcommand{\Qcircuit}{\xymatrix @*=<0em>}

\spnewtheorem{protocol}[theorem]{Protocol}{\bfseries}{\itshape}

\let\emptyset\varnothing

\usepackage{tikz}
\usepackage{tikzpeople}
\usepackage{stackengine}
\usepackage{relsize}
\usepackage{setspace}
\usepackage{pdfpages}
\usepackage{marvosym}

\title{Non-Locality and Zero-Knowledge MIPs}

\author{
Claude Cr\'epeau \inst{1}\! \thanks{Supported in part by FRQNT (INTRIQ) and NSERC (CryptoWorks21 and Discovery grant program).}
\and 
Nan Yang\inst{2}\! \thanks{Supported in part by Professors V\'aclav~Chv\'atal, Jeremy~Clark, Claude~Cr\'epeau, and David~Ford.}
}

\institute{
McGill University, Montr\'eal, Qu\'ebec, Canada.
{crepeau@cs.mcgill.ca}
\and
Concordia University, Montreal, Quebec, Canada.
{na\_yan@encs.concordia.ca}
}

\begin{document}

\maketitle

%
%
%
%
%
%


\begin{abstract}




The foundation of zero-knowledge is the \emph{simulator}: a weak machine capable of pretending to be a weak verifier talking with all-powerful provers. To achieve this, simulators need some kind of advantage such as the knowledge of a trapdoor. In existing zero-knowledge multi-prover protocols, this advantage is essentially \emph{signalling}, something that the provers are explicitly forbidden to do. This advantage is stronger than necessary, as it is possible to define a sense in which simulators need much less to simulate. We define a framework in which we can quantify the simulators' \emph{non-local advantage} and exhibit examples of zero-knowledge protocols that are sound against local or entangled provers that are not sound against no-signalling provers precisely because the no-signalling simulation strategy can be adopted by malicious provers.
\end{abstract}

\section{Introduction}

An \emph{interactive proof} is a dialog between two parties: a polynomial-time \emph{verifier} and an all-powerful \emph{prover} \cite{GoldwasserMiRa89,Babai85}. They agree ahead of time on some language $L$ and a string $x$. The prover wishes to convince the verifier that $x \in L$. If this is true, the prover should succeed almost all the time; if not, the prover should fail almost all the time. This is a generalization of the complexity class $\mathbf{NP}$, except instead of simply being handed a polynomial-sized witness, the verifier is allowed to quiz the prover. The set of languages that admit an interactive proof is called $\mathbf{IP}$.

An interactive proof is \emph{zero-knowledge} if the verifier learns nothing except the truth of ``$x \in L$''. This is usually defined by saying that a \emph{distinguisher} is unable to tell apart a real conversation between the prover and the verifier, and one which is generated by a lone polynomial-time \emph{simulator}. We will denote sets of zero-knowledge interactive proofs with a $\mathbf{ZK}$ bold prefix.

The \emph{multi-prover} model was introduced in \cite{BGKW88}. This model consists of multiple, non-commu\-nicating\footnote{The precise meaning of these words shall become a lot clearer throughout the rest of this paper.} provers talking to a single verifier. We will abbreviate ``multi-prover interactive proof'' as MIP and the set of languages which can be accepted by MIPs as the boldface $\mathbf{MIP}$.

From a \emph{complexity} perspective, the zero-knowledge aspect of interactive proofs is characterized by $\mathbf{IP} = \mathbf{ZKIP}\footnote{This is for computational Zero-Knowledge. For statistical ZK however the corresponding class is $\mathbf{SZK}$ and is most likely contained in $\mathbf{MA}$, contains $\mathbf{BPP}$ and most likely contains only some part of $\mathbf{NP}$.} = \mathbf{PSPACE}$ for single-prover IPs (\cite{Shamir:1992:IP,ImpagliazzoY87,Ben-Or:1988}), and $\mathbf{MIP} = \mathbf{ZKMIP} = \mathbf{NEXP}$ for multi-prover IPs (\cite{BGKW88,Fortnow:1994:PMI:194527.194556,BFL90,KILIAN89,DBLP:conf/stoc/FeigeK94,DBLP:conf/crypto/DworkFKNS92,DBLP:journals/siamcomp/FeigeK00}). The (conjectured) necessity of complexity assumptions for zero-knowledge in the single-prover case was the initial motivation for the multi-prover model.

\subsection{A Cryptographic Perspective}

The foundation of zero-knowledge is the idea of a \emph{simulator}: a machine, with no more power than the verifier, which can pretend to having interacted with all-powerful provers. Obviously, this simulator cannot accomplish this task without some kind of \emph{advantage} (something independent of knowledge). In single-prover zero-knowledge proofs, this advantage can be in the form of the ability to \emph{rewind} computation, the ability to discard failed simulations, or knowledge of a trapdoor in a commitment scheme. In multi-prover zero-knowledge proofs, the advantage in existing literature can be summed up as \emph{signalling}: the simulator, acting in the name of several provers, knows secrets which real provers, in a real instance of the protocol, would not because they are unable to communicate.

From a complexity perspective, this simulator advantage can be anything as long as it is truly independent of knowledge -- we do not want to exclude anything a priori. But, in practice, zero-knowledge is ultimately applied cryptography and from a cryptographic perspective, \emph{not all advantages are equal}.

\subsection{Relativistic Motivation}

The need for more nuanced simulators is motivated by relativistic cryptography, an example of which can be found in \cite{PhysRevLett.115.030502}. Relativistic cryptography exploits the fact that it is impossible to signal faster than light. We can enforce the no-signalling condition of MIPs by spatially separating the provers from each other. In order to enforce the provers' spatial separation during the execution of the protocol, each prover is paired with a verifier of its own, which is located nearby. The verifiers can use the timing of the replies of their respective provers to judge their relative distance.

In practice, this means that we can implement MIPs under relativistic assumptions if the verifier can be ``split'' into multiple verifiers, each locally interacting with its corresponding prover. An example of relativistic cryptography can be found in \cite{PhysRevLett.115.030502}, where a commitment was sustained for over 24 hours.



Some MIPs have verifiers which, intrinsically, cannot be split. Examples include \cite{BGKW88} and \cite{KILIAN89}. In these examples, the verifier is used to courier an authenticated message between provers. In the relativistic setting, if the verifier has time to pass a message between provers, then the provers just signal between themselves.

Luckily, most MIPs in the literature have verifiers that are \emph{non-adaptive}. These verifiers' questions to one prover are independent of the answers from all the provers. MIPs with non-adaptive verifiers can be rewritten into a format with multiple, split verifiers; this format we will call \emph{locality-explicit}, and will be defined formally in section \ref{SEC:CCMIP}.

As an example of what we mean, consider the following two-prover interactive proof for graph 3-coloring:

\begin{quote}
\rule{\linewidth}{1pt}
\begin{protocol}\label{PROT:INTRO_3COLOR_A} ( Simple MIP, Single-Verifier ) \end{protocol}

Two provers $P_1, P_2$, one verifier $V$. On input graph $G$, $P_1$ and $P_2$ agree on a 3-coloring.

\begin{enumerate}

\item $V$ asks $P_1$ for the colors of an edge $e$.

\item $V$ asks $P_2$ for the colors of one of the nodes of $e$.

\end{enumerate}

$V$ accepts if and only if the colors of that edge from $P_1$ are not equal, and $P_2$ corroborates with $P_1$'s answer by replying with the same color for the same node.

\rule{\linewidth}{1pt}
\end{quote}

In the above protocol, $V$'s questions to either prover does not depend on answers from any prover. This is what is commonly known as a {\em non-adaptive} verifier. We can therefore split the above verifier into a two-verifier version:

\begin{quote}
\rule{\linewidth}{1pt}
\begin{protocol}\label{PROT:INTRO_3COLOR_B} ( Simple MIP, Multi-Verifier ) \end{protocol}

Two provers $P_1, P_2$, two verifiers $V_1, V_2$. On input graph $G$, $P_1$ and $P_2$ agree on a 3-coloring, $V_1$ and $V_2$ agree on an edge $e$.

\begin{enumerate}

\item $V_1$ asks $P_1$ for the colors of $e$.

\item $V_2$ asks $P_2$ for the colors of one of the nodes of $e$.

\end{enumerate}

Post execution, $V_1$ and $V_2$ confer with each other, and accept if and only if the colors of that edge from $P_1$ are not equal, and $P_2$ corroborates with $P_1$'s answer by replying with the same color.

\rule{\linewidth}{1pt}
\end{quote}

This version of the protocol is naturally suited for relativistic implementation. However, it is not zero-knowledge because even if $P_1$ and $P_2$ agreed on a randomly selected 3-coloring each time,  a dishonest verifier
$V_{2}$ may sample a node which is not from $e$.
We can make a zero-knowledge, multi-verifier MIP with the help of the following commitment scheme, which is adapted from \cite{BGKW88}:
\pagebreak
\begin{quote}
\rule{\linewidth}{1pt}
\begin{protocol}\label{PROT:INTRO_BC} ( Multi-Verifier Commitment ) \end{protocol}

Two provers $P_1, P_2$, two verifiers $V_1, V_2$. The provers share a random string $w$, and the verifiers share a random string $r$. Operations are over a finite field. $P_1$ wishes to commit $b$.

\begin{enumerate}

\item (Commit) $V_1$ sends $P_1$ the string $r$. $P_1$ replies with $x = w + br$.

\item (Unveil) $P_2$ sends $V_2$ the string $w$.

\end{enumerate}

Post execution, the verifiers confer. They accept if and only if $x + w = r$ or $x + w = 0$.

\rule{\linewidth}{1pt}
\end{quote}

Combining protocol \ref{PROT:INTRO_BC} and the zero-knowledge protocol of \cite{Goldreich:1991:PYN:116825.116852} gives us a zero-knowledge, multi-verifier MIP.

\begin{quote}
\rule{\linewidth}{1pt}
\begin{protocol}\label{PROT:INTRO_3COLOR_C} ( ZKMIP, Multi-Verifier ) \end{protocol}

Two provers $P_1, P_2$, two verifiers $V_1, V_2$. On input graph $G$, $P_1$ and $P_2$ agree on a randomly selected 3-coloring and $2|V|$ strings $w_{i}$, $V_1$ and $V_2$ agree on an edge $e$ and $2|V|$ strings $r_{i}$.

\begin{enumerate}

\item $P_1$ commits the colouring of $G$ to $V_1$ using the $2|V|$ $w_{i}$, $r_{i}$ they pre-agreed.

\item $V_2$ asks $P_2$ to unveil the colours of the edge $e$.

\end{enumerate}

Post execution, $V_1$ and $V_2$ confer with each other, and accept if and only if the commitment is valid, and the colors unveiled are not equal.

\rule{\linewidth}{1pt}
\end{quote}

What makes this protocol zero-knowledge? In the commitment scheme (protocol \ref{PROT:INTRO_BC}), if $P_2$ has knowledge of $r$, then it can break the commitment by unveiling either way (by sending $w$ or $w + r$ as needed). Following the precedents set by existing literature's definition of zero-knowledge, the (\emph{single}) simulator, interacting with both verifiers, learns $r$. Therefore it can break the commitment and always unveil a color that will be accepted by the verifiers.

\subsection{Simulator's Advantage}

As mentioned, the (single) simulator's advantage is its ability to interact with both verifiers at once. This is equivalent to having a pair of simulators signaling and, as we will see, is actually a tremendous power. However, it turns out that simulators do not need to signal in order to break the above commitment (section \ref{SEC:Standard}); a weaker non-local distribution will do. What we wish is to construct a framework in which this ``non-local advantage'' of the simulators can be quantified. We do this in section \ref{SEC:CCMIP}.

To see how much overkill signaling is for the simulators, imagine that in the above protocol, the distinguisher were able to eavesdrop on the ``conversation'' between the (possibly malicious) verifiers and black boxes, inside of which are either real provers, or simulators. This is giving the distinguisher more power than simply reading a transcript; and yet, the (signaling) simulators can succeed not only in generating the transcript, but behave as if they were provers in real-time. If we consider existing zero-knowledge as ``transcript-indistinguishable'', then we may consider this as ``eavesdrop-indistinguishable''. We will leave these terms undefined (as intuition) as they are not the focus of this work.

\subsection{Our Contributions}

In this work, we propose a framework for writing MIPs which is naturally suited for implementation and analysis under relativistic assumptions. We discuss how this framework extends naturally to zero-knowledge protocols and quantifies the non-local advantage which simulators use in many ZKMIPs. We show that $\mathbf{NEXP}$ can be accepted by MIPs in this form, and discuss the relationship between simulators' non-local advantage and soundness.

We exhibit a MIP for $\mathbf{NP}$ which, if is zero-knowledge, then cannot be sound; we introduce this as a tool for proving impossibility results of soundness against no-signalling provers but it could be used for
for any non-locality class similarly.

\section{Previous Work}


The early work by Ben-Or, Goldwasser, Kilian and Wigderson asserting that $\mathbf{ZKMIP}=\mathbf{MIP}$ from \cite{BGKW88} and \cite{KILIAN89} use multi-round protocols and their (honest) verifiers are inherently signaling.
This is precisely why we address the situation in this work. Proving soundness is quite subtle in this case because the provers could use the (signaling) verifier to break binding of the commitments. In particular, soundness will not be valid if the protocol is composed concurrently with other executions of itself or even used as a sub-routine.
In recent conversations with Kilian \cite{KILIAN18}, we have learned that controlling the impact of this \emph{signaling} (via the verifier) has been a concern since the early days of MIPs.
The protocols as they are might be sound but it is not fully proven anywhere in writing.
However, it is also clear that no considerations had been given to the fact that general non-local correlations are possible via the verifier. If soundness rests on the binding property of a commitment scheme (such as those zero-knowledge proofs) and this binding property rests on the inability to achieve a certain non-local correlation then impossibility to achieve this correlation via the verifier must be demonstrated.
It is not done or hinted in these papers.

The multi-round issue we address may seem trivial because it is a known fact that multi-round MIPs may be reduced to a single round using techniques
of Lapidot-Shamir \cite{DBLP:conf/focs/LapidotS91} and Feige-Lovasz \cite{Feige:1992:TOP:129712.129783}.
Nevertheless, if interested in \emph{zero-knowledge} MIPs, commitment schemes are generally used to obtain the zero-knowledge property and thus the single-round structure is lost in the process.
Although single-round protocols bypass verifier's non-local contamination problems we describe in this work, converting multi-round protocols into single-round ones is highly inefficient and complex.
Preserving zero-knowledge while achieving single-round has turned out to be a major challenge.
Practically, keeping a multi-round protocol's structure, using only commitments to achieve zero-knowledge is very appealing.

In \cite{DBLP:conf/focs/LapidotS91}, Lapidot-Shamir proposed a parallel ZKMIP for $\mathbf{NEXP}$, but they removed the zero-knowledge claim
in the journal version \cite{DBLP:journals/jcss/LapidotS97} of their work without any explanation as of why.
Feige and Kilian \cite{DBLP:conf/stoc/FeigeK94} were the last ones to follow this approach combining techniques drawn from Lapidot-Shamir \cite{DBLP:conf/focs/LapidotS91},
Feige-Lovasz \cite{Feige:1992:TOP:129712.129783} and Dwork, Feige, Kilian, Naor, and Safra, \cite{DBLP:conf/crypto/DworkFKNS92} to achieve a ``2-prover 1-round 0-knowledge'' proof for $\mathbf{NEXP}$.
As far as we can tell, this is the only paper in the ZKMIP literature that appears to avoid the multi-round problems and the non-local contamination that we discuss.
However, note that the analysis of \cite{DBLP:conf/stoc/FeigeK94} is partly based of that of \cite{DBLP:conf/focs/LapidotS91}, and the journal version of
Feige-Kilian \cite{DBLP:journals/siamcomp/FeigeK00} does not contain their prior claim of zero-knowledge either.
All other ZKMIPs for $\mathbf{NEXP}$ in the literature are multi-round, and thus our analysis applies to them.

Similar issues are possible using more recent results such as Ito-Vidick's proof \cite{IV12} that $\mathbf{NEXP}\subseteq \mathbf{MIP}^{*}$,
Kalai, Raz and Rothblum's proof \cite{KRR14} that $\mathbf{MIP}^{ns}\! = \mathbf{EXP}$ and Natarajan-Wright's proof \cite{NW19} that $\mathbf{NEEXP}\subseteq \mathbf{MIP}^{*}$.
The reason why these multi-round constructions may maintain their soundness despite the potential non-locality contamination (via the verifier) is the {\em non-adaptive} nature of their verifiers.
Non-adaptive verifiers cannot take advantage of information acquired in recent rounds to construct new questions to the provers: all their questions are pre-established before the interaction
with the provers start. This is a special simpler case of local verifiers. Nowhere in this large literature can one find a single statement observing the non-adaptiveness of the verifiers and
its importance to guarantee soundness of those MIPs.
Moreover, their multi-round structure requires that any straightforward extensions to $\mathbf{ZKMIP}^{*}$ or $\mathbf{ZKMIP}^{ns}$ via commitment schemes be analyzed very carefully
and the locality of the resulting verifiers be re-established. This is part of the reasons why the ZK version did not follow easily.
Recently, Chiesa, Forbes, Gur, and Spooner \cite{DBLP:journals/eccc/ChiesaFGS18} discovered a proof that $\mathbf{NEXP}\subseteq \mathbf{ZKMIP}^{*}$.
Their construction is based on refinements of Ito-Vidick's proof and along the lines of Feige-Kilian, building on algebraic structures to bypass the need of commitment schemes.
Unfortunately, this work is so complicated that we are unable to assess whether their verifier is actually non-adaptive.
And of course, this is not mentioned or proven anywhere nor available from the authors...
At the time of writing this paper, we just found out that indeed $\mathbf{ZKMIP}^{*} = \mathbf{MIP}^{*}$ as proven by Grilo, Slofstra and Yuen \cite{GSY19}.

Bellare, Feige, and Kilian \cite{DBLP:conf/istcs/BellareFK95} considered a multi-verifier model similar to ours in order to analyze the role of randomness in multi-prover proofs.
This is completely unrelated to our goal of analyzing verifier non-local contamination.
Finally, the notion of relativistic commitment schemes put forward by Kilian \cite{kilian1990strong} and Kent \cite{PhysRevLett.83.1447} leads to several results \cite{PhysRevLett.115.030502,DBLP:journals/corr/AdlamK15,CL17} where a similar multi-verifier model is necessary in order to assess spatial separation of the provers.
The new (non-local) zero-knowledge definition is 100\% fresh from this work. No prior work exists at all.

\section{The Standard MIP Model}\label{SEC:Standard}

Multi-prover interactive proofs were introduced in \cite{BGKW88}. The intuition for their model was that of a detective interrogating two suspects held in different rooms. This was formalized as follows:

\begin{definition}

Let $P_1, \ldots, P_k$ be computationally unbounded Turing machines and let $V$ be a probabilistic polynomial-time TM. All machines have a read-only input tape, a read-only auxiliary-input tape, a private work tape and a random tape. The $P_i$'s share a joint, infinitely long, read-only random tape. Each $P_i$ has a write-only communication tape to $V$, and vice-versa. We call $(P_1, \ldots, P_k, V)$ a \emph{$k$-prover IP}, or multi-prover interactive proof (MIP).

\end{definition}

This model is essentially equivalent to that of Bell \cite{BELL64} who introduced his famous Bell's inequality to distinguish {\em local} parties from {\em entangled} parties.

Zero-knowledge MIPs were also defined in \cite{BGKW88}:

\begin{definition}\label{DEF:standardZK}

Let $(P_1, \ldots, P_k, V)$ be a k-prover IP for language $L$. Let $\mathbf{view}(P_1, \ldots, P_k, V, x)$ denote the verifier's incoming and outgoing messages with the provers, and
his coin tosses\footnote{We ignore auxiliary inputs because we are not going to discuss composition.}.
We say that $(P_1, \ldots, P_k, V)$ is \emph{perfect zero-knowledge} {for} $L$ if there exists an expected polynomial-time machine $M$ such that for all $V'$, $\mathbf{view}(P_1, \ldots, P_k, V', x)$ and $M(x)$ are identically distributed.

\end{definition}

Let us call the above two definitions the \emph{standard MIP model}. There have also been augmentations of the model by giving the {provers} various non-local resources, such as entanglement \cite{IV12}, or arbitrary no-signaling power \cite{KRR14}. 

Of specific interest to us are standard MIPs which have verifiers that are non-adaptive.

\begin{definition}\label{DEF:nonAdaptive}

A verifier is \emph{non-adaptive} if the verifier's questions depend only on its random coins and the input $x$. A MIP with a non-adaptive verifier is a \emph{non-adaptive MIP}.

\end{definition}

Some zero-knowledge MIPs such as \cite{KILIAN89} \emph{require} that the verifier courier an authenticated message between the provers in order to obtain soundness while ensuring zero-knowledge. The gist of it goes like this:

\begin{enumerate}

\item $V$ asks $P_1$ some questions.

\item $V$ wants to check one of $P_1$'s answers with $P_2$ for consistency.

\item In order for zero-knowledge to hold, $V$ \emph{must} ask $P_2$ a question it has already asked $P_1$.

\item $P_1$ authenticates a question with a key that was committed at the beginning of the protocol and sends it to $V$.

\item $V$ sends the question and the authentication to $P_2$, who proceeds only if it succeeds.

\end{enumerate}

Steps 4 and 5 consists of $V$ sending a message from $P_1$ to $P_2$. This is problematic under relativistic assumptions, as discussed in the introduction. Therefore, the no-signaling assumption of standard MIPs are not immediately compatible with the no-faster-than-light-signaling assumption of relativity.

\section{Locality-Explicit MIP}\label{SEC:CCMIP}


We define a framework for writing MIPs guaranteeing compatibility with relativistic assumptions. This framework uses multiple verifiers, each of which talks to a single prover; in turn, each prover talks to that single verifier. There are no communication tapes between the verifiers, nor are there between provers. There is a special verifier $V_{0}$ which \emph{only reads} the outputs of the other verifiers; this is the verifier that will decide to accept or reject membership to $L$. We call this model ``locality-explicit'' since the provers and verifiers are explicitly local.

Any correlational resources available are explicitly specified via a supplementary {\em correlator} named $\widehat{P}$ for the provers and $\widehat{V}$ for the verifiers. Examples of these resources include entanglement, no-signalling distributions, or slower-than-light signalling.





\begin{definition}

An \emph{interactive Turning machine} (ITM) is augmented with the following tapes:

\begin{itemize}

\item $k_1$ read-only incoming communication tapes.

\item $k_2$ write-only outgoing communication tapes.

\item Private work, auxiliary-input, and random tapes.

\end{itemize}
An ITM $A$ can signal to ITM $B$ if $A$'s write-only outgoing tape is $B$'s read-only incoming tape.

\end{definition}

\begin{definition}
Let $(\widehat{P}, P_1, \ldots, P_k, \widehat{V}, V_0, V_1, \ldots, V_k)$ be a tuple of ITMs, where the $P\!$'s are computationally all-powerful and the $V\!$'s are polynomial-time. 
For each $i$, there are two-way communication tapes between $V_i$ and $P_i$, and that for all $j$, there is a two-way communication tape between $\widehat{V}$ and $V_j$ and also between $\widehat{P}$ and $P_j$. In addition, for each $\ell$, there is a read-only tape going from $V_\ell$ to $V_0$ (where $V_0$ reads). Then, this is said to be a \emph{locality-explicit multi-prover interactive proof}.

We call $\widehat{P}$ and $\widehat{V}$ \emph{correlators} and say that the provers and verifiers are $\widehat{P}$-\emph{local} and $\widehat{V}$-\emph{local} respectively. We define the class of all MIPs with such correlators $\mathbf{MIP}^{\widehat{P}}_{\widehat{V}}$.

\end{definition}

It is perhaps easier to understand our definition with the help of figure \ref{CCMIP}.

\begin{figure}[h]
\centering
\includegraphics[angle=0, width=0.7\textwidth]{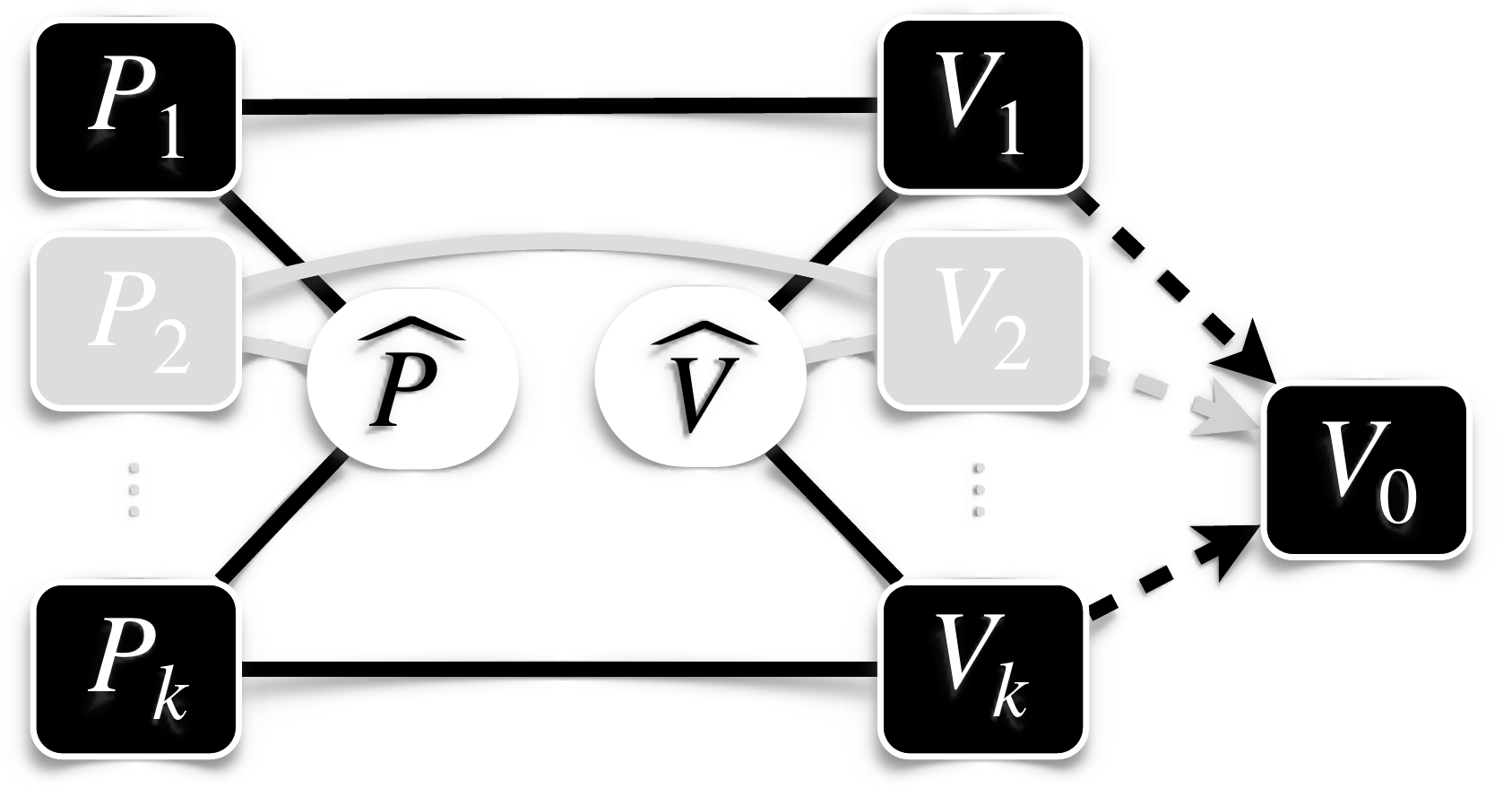}
\caption{\label{CCMIP} Locality-Explicit MIP}
\end{figure}

The solid lines represents two-way communication and the dashed arrows represents one-way communication, with the arrow indicating the direction of information flow. 

We can define that an LE-MIP accepts a language $L$ if the usual soundness and completeness conditions hold:

\begin{definition}

An LE-MIP $(\widehat{V}, V_0, V_1, \ldots, V_k, \widehat{P}, P_1, \ldots, P_k)$ accepts a language $L$ if and only if

\begin{itemize}

\item (completeness) $\forall x \in L, \mathbf{Pr}[V_0(x, t_1, \ldots, t_k) = \emph{\textsf{accept}}] > 2/3$,

\item (soundness) $\forall x \notin L, \forall P_1^{\prime}, \ldots, P_k^{\prime}, \mathbf{Pr}[V_0(x, t_1, \ldots, t_k) = \emph{\textsf{accept}}] < 1/3$,

\end{itemize}
where $t_i$ is the read-only tape from $V_i$ to $V_0$ at the end of $V_i$'s interaction with $P_i$ (or $P_i^{\prime}$) on input $x$.

\end{definition}

Note that we do not quantify over $\widehat{P}$ (nor $\widehat{V}$), as we want to use them not as (possibly malicious) participants to the protocol, but as a description of correlational resources available to the provers and verifiers.


\begin{definition}\label{DEFMIP}

An LE-MIP is \emph{local} if $\widehat{V} = \widehat{P} = \emptyset$ and all of the provers' (resp.~verifiers') random tapes are initialized with the same uniformly random string $R$ (resp.~verifiers with another, independent uniformly random string $S$)\footnote{By $\emptyset$ we mean the empty correlator that provides everyone with nothing at all as output whatever the input is.}.

\end{definition}

MIPs in the standard model (with local provers) are equivalent to LE-MIPs where $\widehat{P} = \emptyset$ and $\widehat{V}$ acts as a bulletin board. That is, a single verifier communicating with multiple provers is equivalent to multiple verifiers individually communicating with a local prover and each among themself.

\begin{lemma}

If a MIP is non-adaptive, then there exists a local LE-MIP which accepts it.

\end{lemma}

This is obvious as a non-adaptive verifier's questions are decided ahead of time, once its random coins are fixed. Therefore, we may split the verifier into one for each prover with a list of predetermined questions.



%
%
%
%
%

\subsection{Zero-Knowledge LE-MIPs}\label{SEC:NEWZK}


%

As discussed in the introduction, zero-knowledge is defined by simulations. The simulator of single-prover IP and standard MIP are equal to the verifier in computational power, but they do have ``advantages'' -- such as the ability to rewind computation.

LE-MIPs makes explicit a new advantage for the simulator: non-local correlations, a very powerful advantage. Using the correct non-local correlations, simulators do not need to rewind, do not need to pretend to be multiple (isolated) provers, and do not need to know any commitment-breaking secrets. In short, they do not need to signal. Multiple, no-signaling simulators can even produce transcripts in ``real-time'' (example will follow) if the proper correlations are used.


%
%

\begin{definition}

Let $\mathcal{M} = (\widehat{M}, M_1, \ldots, M_k)$ be a tuple of polynomial-time ITMs. Each machine has a random tape, and every random tape is initialized with the same random bits. For $1 \leq i \leq k$, there is a two-way communication tape between $\widehat{M}$ and $M_i$. There are no communication tapes between any of the $M_i\!$'s. Then this is called a tuple of \emph{locality-explicit simulators} and $\widehat{M}$ is the \emph{locality class} of $\mathcal{M}$, which will be abbreviated $\widehat{M}$\emph{-local}.

\end{definition}

\begin{definition}

Let $\mathcal{PV}=(\widehat{P}, P_1, \ldots, P_k, \widehat{V}, V_0, V_1, \ldots, V_k)$ be an LE-MIP for language $L$. If there exists a tuple of locality-explicit simulators $(\widehat{S}, S_1, \ldots, S_k)$,
such that for all verifiers $(\widehat{V}',V^\prime_0, V^\prime_1, \ldots, V^\prime_k)$, such that for all $x\in L$ the transcripts of conversations $$(\widehat{P}, P_1, \ldots, P_k, \widehat{V}', V^\prime_0, V^\prime_1, \ldots, V^\prime_k)(x)$$
and those generated by $$(\{ \widehat{S}, \widehat{V}^\prime \}, V^\prime_0, S_1^{V^\prime_{1}}, \ldots, S_k^{V^\prime_{k}})(x)$$ are identically distributed, then we say that $\mathcal{PV}$ 
is a \emph{$\widehat{S}$-local perfect zero-knowledge LE-MIP} for $L$. Note that the simulators are responsible for using $\widehat{V}^\prime$, if necessary, to ensure that the verifier oracles\footnote{Each simulator $S_{i}$ is restricted to oracle calls to its own corresponding $V^\prime_{i}$.} receive the necessary inputs.

We will denote the set of all ZK LE-MIPs where the provers, verifiers, and simulators are $\widehat{P}$-local, $\widehat{V}$-local, and $\widehat{S}$-local by $$\mathbf{ZK}^{\widehat{S}}\mathbf{MIP}^{\widehat{P}}_{\widehat{V}}.$$

Let $\mathbb{S, P, V}$ be sets of correlators. We will denote, by convention, $$\mathbf{ZK}^{\mathbb{S}} \mathbf{MIP}^{\mathbb{P}}_{\mathbb{V}}$$ as the set of all ZK LE-MIPs where each correlator comes from each of the respective sets.

\end{definition}

Our motivations for the above definition are twofold.

First, a simulator (or simulators) should not have more power than necessary. If two \emph{local} simulators can output for two \emph{local} verifiers, then it is not necessary to have a single simulator (equivalent to two \emph{signaling} simulators) do the job. In general, finding the minimal $\widehat{S}$ that will allow simulation establishes how little extra is needed to obtain the zero-knowledge property.


Second, the non-locality of simulators is a characterization of the resilience of zero-knowledge. A protocol with local simulators which can withstand arbitrary (malicious) verifiers is more resilient than one in which signaling simulators are needed.

This may be of practical interest, if transcripts are timestamped. For example, under the relativistic assumption that one may not signal faster-than-light, one may be able to distinguish two spatially separated simulators from two spatially separated verifiers, if the simulators need to signal (transmit a commitment-breaking secret) in order to generate a transcript. On the other hand, if two entangled simulators are sufficient to produce the transcript, then they are indistinguishable from real verifiers and provers. Our protocol \ref{IsolatingZKMIP} can be modified as to let entangled simulators do their work, without needing PR-boxes or signaling. Details in section \ref{SEC:NEXP}


%
%
%
%
%
%
%
%
%

The complexity of LE-MIPs are the same as those of MIP, namely:

\begin{theorem}

There exists a LE-MIP which accepts $\mathbf{NEXP}$.

\end{theorem}

The proof is a line-by-line inspection of the BFL protocol as found in \cite{BFL90}, and checking that the verifier is non-adaptive, and therefore can be written as a LE-MIP. We have included a brief summary of the BFL protocol in appendix \ref{SEC:APP_BFL}.

\section{Zero-Knowledge LE-MIP for $\mathbf{NEXP}$}\label{SEC:NEXP}

The question which follows naturally is whether there exists a \emph{zero-knowledge}, local LE-MIP for $\mathbf{NEXP}$ where $\mathbb{S} {\,\,}_{\not}\!\!\!\subseteq \mathbf{SIG}$. By adapting the protocol from \cite{BFL90}, we will exhibit a protocol with the following properties:

\begin{enumerate}

\item The provers and verifiers are local: $\widehat{V} = \widehat{P} = \emptyset$.

\item The simulators need only access to instances of $\mathbf{PR}$-boxes to work. That is, $\widehat{S}$ simply computes indexed instances of $\mathbf{PR}$-boxes. We will abbreviate this as ``$\mathbf{PR}$-local.''

\end{enumerate}

We may succinctly summarize the above as:

\begin{theorem}\label{THM:zklemip}

$\mathbf{ZK}^{\mathbf{PR}}\mathbf{MIP}^{\emptyset}_{\emptyset} = \mathbf{NEXP}$, where $\mathbf{PR}$ denotes a correlator which simply computes $\mathbf{PR}$-boxes for the simulators.

\end{theorem}

We prove the above theorem by constructing an LE-MIP with the right properties: protocol \ref{IsolatingZKMIP}.
The generic way of turning an interactive proof into a zero-knowledge one is by running it in committed form \cite{BGKW88,KILIAN89}. With this technique, provers commit their answers instead of directly responding, and use cryptographic techniques to convince the verifier that the answers are correct. As argued previously, this is not possible to enforce from relativistic assumptions alone.

Our solution essentially asks the provers to (strongly-universal-2) hash the selected committed answer with a key that is based on the verifier's question. We force $V_2$ to behave honestly (to ask a question that $V_1$ has asked) by making bad questions meaningless. If the verifiers ask the provers the same question, they will receive the same hash of the same answer. Otherwise, they will receive two independent random hash values.

The $\mathbf{PR}$-type commitment (protocol \ref{MIP_CHSHBC}) is secure in the local setting as previously proved in \cite{PhysRevLett.83.1447,CSST11,PhysRevLett.115.030502}. It is perfectly concealing and statistically binding.
In general, we use the commitment-box notation \fbox{$b$} as the name of a commitment to bit $b$ in the next two protocols.
\begin{quote}
\rule{\linewidth}{1pt}
\begin{protocol}\label{MIP_CHSHBC}\label{MIP_CHSH} A statistically binding, perfectly concealing commitment protocol to bit $b$.\end{protocol}

All parties agree on a security parameter $1^{k}$.\\ 
$P_1$ and $P_2$ partition their private random tape into two $k$-bit strings $w_{1},w_{2}$.\\

{\bf Pre-computation phase:}

\begin{itemize}

\item $V_{1}$ samples two $k$-bit strings $z_{1},z_{2}$ independently and uniformly, and provides them to $V_{2}$.

\item $V_{1}$ sends $z_{1}$ to $P_1$ and $V_{2}$ sends $z_{2}$ to $P_2$.

\end{itemize}

{\bf Commit phase:}

\begin{itemize}

\item $P_1$  commits $b$ to $V_{1}$ as \fbox{$b$} $ = (b \times z_{1})  \oplus  w_{1}$, where $b \times z_{1}$ is a multiplication in $\mathbb{F}_{2^{n}}$.

\item $P_2$ sends $V_{2}$: $d = (w_{1} \times z_{2}) \oplus w_{2}$.

\end{itemize}

{\bf Unveiling phase:}

\begin{itemize}

\item $P_1$ sends $w_{1},w_{2}$ to $V_{1}$.

\item $V_{1}$ computes $b = 1$ if $\text{\fbox{$b$}} \oplus  w_{1} = z_{1}$, or $b =0$ if $\text{\fbox{$b$}} =  w_{1}$. 

\item $V_{0}$ {\bf rejects} if $\text{\fbox{$b$}} \oplus w_{1}$ is anything but $z_{1}$ or $0$, or if $d \oplus w_{2} \neq w_{1} \times z_{2}$
and {\bf accepts} $b$ otherwise.

\end{itemize}

\rule{\linewidth}{1pt}
\end{quote}


A note on notation: for a circuit $f$, we will denote $f\!\left( \text{\fbox{$x$}} \right)$ as the gate-by-gate committed circuit evaluated with x as the input.
We also use statements such as ``$P_1$ proves to $V_1$ that \text{\fbox{$\Omega_1$}} was computed correctly''. The reader is expected familiarity with zero-knowledge computations on committed circuits as put forward by \cite{BrassardCa86,BrassardCb86,ImpagliazzoY87,KILIAN89}.

\begin{quote}
\rule{\linewidth}{1pt}
\begin{protocol}\label{IsolatingZKMIP} A local zero-knowledge LE-MIP for oracle-3-SAT\end{protocol}

Let $x=(B,r,s)$, an instance of oracle-3-SAT, be the common input, let $k = \left| x \right| = r+3s+3$, and let $\Lambda$ be the verifier's program in protocol \ref{SEC:APP_BFL} (see appendix).

\begin{enumerate}

\item {\bf Pre-computation:}

	\begin{enumerate}

	\item $V_{1}$ samples two $k$-bit strings $z_{1},z_{2}$ independently and uniformly, and provides them to $V_{2}$.

	\item $V_{1}$ selects $k+3$ random bit strings $R_1,...,R_{k+3}$ (size specified implicitly by $\Lambda$) and evaluates the circuit of $\Lambda$ using the $R_i$ as randomness, resulting in questions $Q_1,...,Q_{k+3}$, and provides them to $V_{2}$

	\item $V_{1}$ randomly chooses $i$, $1 \leq i \leq k+3$, the index of an oracle query that will be made to both $P_1$ and $P_{2}$. $V_{1}$ provides $i$ to $V_{2}$.

	\item $V_{1}$ sends $z_{1}$ to $P_1$ and $V_{2}$ sends $z_{2}$ to $P_2$ for future commitments.
	
	\item All parties agree on a family of strongly-universal-2 hash functions $\{H_i\}$ indexed by $k$-bit keys.
	
	\item $P_1$ and $P_2$ agree on a $k$-bit index $\gamma$ to the above family. $P_1$ commits \fbox{$\gamma$} to $V_{1}$.

	\end{enumerate}

\item {\bf Sumcheck with oracle:}

	\begin{itemize}

	\item Let $f$ be the arithmetization obtained in protocol \ref{PROT:BFL_sumcheck}, 
let $z$ be a string from $I^r$ and $Q_{k+1}, Q_{k+2}, Q_{k+3}$ be strings of $I^s$ as generated in protocol \ref{SEC:APP_BFL}. 
$V_{1}$ and $P_1$ execute protocol \ref{PROT:BFL_sumcheck} in committed form.
At the end of this phase, $P_1$ shows that the committed final value is equal to
$$f\!\left( z, Q_{k+1}, Q_{k+2}, Q_{k+3}, \text{\fbox{$A(Q_{k+1})$}}, \text{\fbox{$A(Q_{k+2})$}}, \text{\fbox{$A(Q_{k+3})$}} \right), $$
an evaluation in committed form of $f$ using the committed values that were used during the protocol's loop. If this fails, $V_{1}$ instructs $V_{0}$ to reject.

	\end{itemize}

\item {\bf Multilinearity test:}

\begin{enumerate}

	\item For $1 \leq i \leq k$:

	\begin{enumerate}
		\item $V_1$  sends $Q_{i}$ to $P_{1}$,
		\item $P_1$ commits his answer as \fbox{$A(Q_i)$}.
	\end{enumerate}
	
	\item $P_1$ and $V_{1}$ evaluate a circuit description of $\Lambda$ in committed form with inputs $\text{\fbox{$A(Q_1)$}}, \ldots, \text{\fbox{$A(Q_k)$}}$ to verify proper linearity among them.
$P_1$ unveils the circuit's committed output. If it rejects, $V_{1}$ instructs $V_{0}$ to reject.

	\end{enumerate}

\item {\bf Consistency test:}

	\begin{enumerate}

	\item $V_1$ sends $i$ to $P_1$.
	
	\item $P_1$ computes $\text{\fbox{$\Omega_1$}} = \text{\fbox{$A(Q_i)$}} \oplus H_{\text{\fbox{$\gamma$}}}\!\left( Q_{i}  \right)$ and sends \text{\fbox{$\Omega_1$}} to $V_1$.
	
	\item $P_1$ proves to $V_1$ that \text{\fbox{$\Omega_1$}} was computed correctly, from the existing commitments.
	
	\item $P_1$ unveils \text{\fbox{$\Omega_1$}} for $V_1$, who gets $\Omega_1$.

	\item $V_{2}$ sends $Q_{i}$ to $P_2$ (recall that this was pre-agreed in step 1.(c))

	\item $P_2$ responds to $V_2$ with $\Omega_2 = A(Q_i) \oplus H_{\gamma}\!\left( Q_{i}  \right)$. 

	\item $V_{0}$ accepts if and only if all of the following conditions are met:
	
		\begin{itemize}
	
		\item $\Omega_1 = \Omega_2$
		\item All commitments which have been unveiled are valid.
		\item $V_{1}$ did not reject in the two previous cases.
		
		\end{itemize}

	\end{enumerate}

\end{enumerate}

\rule{\linewidth}{1pt}
\end{quote}

The proofs of security can be found in appendix \ref{PofS}.

\subsection{Minimal Simulator Advantage}\label{SEC:entangledsim}

What is the minimal simulator advantage needed for achieving zero-knowledge for $\mathbf{NEXP}$?

It is clear that signalling simulators can succeed in the above protocol. This is the zero-knowledge simulator of standard MIPs. We can summarize this as $$\mathbf{ZK}^{\mathbf{SIG}}\mathbf{MIP}^{\emptyset}_{\emptyset} = \mathbf{NEXP},$$ where $\mathbf{SIG}$ is a signalling correlator.

Signalling is however unnecessary, as the binding condition of commitment used above (protocol \ref{MIP_CHSH}) can be broken given $\mathbf{PR}$-boxes. This is what the proof of security shows in appendix \ref{PofS}. Thus, the simulator's advantage can be lowered to $\mathbf{PR}$-boxes, or $$\mathbf{ZK}^{\mathbf{PR}}\mathbf{MIP}^{\emptyset}_{\emptyset} = \mathbf{NEXP}.$$

If the verifiers were willing to tolerate approximately $15\%$ of errors in the provers' unveiling string ($z_{1}$ or $0$), then it is possible to break binding with shared entanglement \cite{10.1007/978-3-540-45078-8_1} while maintaining soundness against local provers. Making this slight change in the protocol reduces the simulator advantage further: $$\mathbf{ZK}^{\mathbf{ENT}}\mathbf{MIP}^{\emptyset}_{\emptyset} = \mathbf{NEXP},$$ where $\mathbf{ENT}$ denotes polynomial amount of shared entanglement for the simulators.

Ideally, the simulators would not need any non-local advantage over the verifiers. However, we are unable to find a zero-knowledge MIP where the simulators are \emph{local} which can accept $\mathbf{NEXP}$, or prove that it is impossible. We make the following conjecture:

\begin{conjecture}

$\mathbf{ZK}^{\emptyset}\mathbf{MIP}^{\emptyset}_{\emptyset} = \mathbf{SZK}$, where $\mathbf{SZK}$ is the set of languages with statistical zero-knowledge interactive proofs without computational assumptions (i.e., graph isomorphism).

\end{conjecture}

\subsection{Soundness Against No-Signalling Provers}\label{SEC:example}

As a further example of the drastic differences between MIP simulators' non-local advantages and single-prover IP simulators' advantages (e.g., rewinding), consider the following:

\begin{theorem}

Suppose that the provers in protocol \ref{IsolatingZKMIP} have access to PR-boxes (thus they are no-signalling, but not local), then the protocol is not sound.

\end{theorem}

\begin{proof}

The provers adopt the simulators' strategy. Since commitment binding is broken with the aid of PR-boxes, the verifiers will always accept.

\end{proof}

This is the sense to which we referred to as ``eavesdrop indistinguishable'' from ``transcript indistinguishable'' earlier. A prover having the ability to rewind computations, although enough for simulators in IPs, is not enough to break soundness. We will generalize the above theorem in a future work, on the relationship between zero-knowledge and soundness.

\subsubsection{Another Example}

In appendix E a zero-knowledge protocol for $\mathbf{NP}$ is extracted from \cite{CMSSY09}. This protocol is not only sound against local provers but also against entangled provers. It is zero-knowledge in both cases. However, since the ZK simulator (also provided in appendix E) can be implemented as no-signalling simulators, this same protocol cannot be sound against no-signalling provers since they can adopt exactly the simulators' strategy.

\section{Conclusions and Future Work}\label{SEC:conclusions}


Zero-knowledge simulators need advantages in order to function. In the case of MIPs, it was always implicitly assumed this advantage is necessarily signaling. We have shown that this is not true, and that this aspect of zero-knowledge remains unexplored. LE-MIPs make this explicit, while providing a template for relativistic implementations of the no-signaling assumption.

We close with three open questions.

First, although the provers and verifiers of protocol \ref{IsolatingZKMIP} are local, the simulators are not -- they use PR-boxes. We do not know whether it is possible to simulate protocol \ref{IsolatingZKMIP} with \emph{local} simulators. In fact, we conjecture that there does not exist a $\mathbf{ZK}^{\emptyset}\mathbf{MIP}^{\emptyset}_{\emptyset}$ protocol for any language outside $\mathbf{SZK}$.

Second, as we have sketched out in section \ref{SEC:entangledsim}, by weakening the commitment scheme used, we get $\mathbf{ZK}^{\mathbf{ENT}}\mathbf{MIP}^{\emptyset}_{\emptyset} = \mathbf{NEXP}$. What is a minimal $\widehat{S}$ such that $\mathbf{ZK}^{\widehat{S}}\mathbf{MIP}^{\emptyset}_{\emptyset} = \mathbf{NEXP}$?

Third, what is the relationship between zero-knowledge and soundness in MIPs? As we have shown in section \ref{SEC:example}, some simulators' strategy can be adopted by provers to break soundness, if only the provers had some additional (in this case, non-local) resources. Is there a relationship between the non-local resources needed to achieve zero-knowledge and those that are forbidden in order to achieve soundness?



\section*{Acknowledgements}
We would like to thank
G.~Brassard,
A.~Chailloux,
S.~Fehr,
J.~Kilian,
S.~Laplante,
J.~Li,
A.~Leverrier,
A.~Massenet,
S.~Ranellucci,
L.~Salvail,
C.~Schaffner,
and
T.~Vidick
for various discussions about earlier versions of this work. We would also like to thank Jeremy Clark for his insightful comments. Finally, we are grateful to Raphael Phan and Moti Yung for inviting us to publish a lead-up paper to this work as an {\em Insight Paper} at MyCrypt 2016.

\bibliography{bib_MIP}
\bibliographystyle{ieeetr}


\newcommand{\RGRB}[0]{ \text{$\mathbf{R  \! \frac{GR}{BG} \! B}$} }
\appendix

\section{Proofs of Security for Protocol \ref{IsolatingZKMIP}}\label{PofS}

\subsubsection{Locality}~\\


Since the protocol is written as an LE-MIP in which $\widehat{P} = \widehat{V} = \emptyset$, the protocol is local by definition \ref{DEFMIP}.

\subsubsection{Completeness}~\\

Completeness follows from the completeness of the underlying protocol \cite{BFL90}, and the fact that the commitment protocol (protocol \ref{MIP_CHSH}) is well-defined for honest provers (who will never send a commitment that they cannot unveil).

\subsubsection{Soundness}~\\


Without loss of generality, we may assume that the soundness error in the BFL protocol to be $1/3$, through sequential amplification. The probability that our commitment scheme (protocol \ref{MIP_CHSH}) fails binding is exponentially small in $k$.
Local probabilistic provers are equivalent to local deterministic provers. This is because the success probability $\alpha$ of randomized provers of breaking soundness is an average over the randomized provers' random tapes. Each instance of a random tape represents a deterministic strategy. Therefore there is a deterministic strategy which succeeds with probability at least $\alpha$, and hence we only need to consider local deterministic provers.

Since $P_1$ is deterministic, we may unambiguously consider what happens if we were to ``rewind'' the prover machine. Suppose that at some point $P_1$ unveils a particular commitment $c$ to $0$. We rewind $P_1$ and let $V_1$ make different choices before that point.
Suppose that, with these alternate choices, $P_1$ then unveils $c$ to $1$ (an attempt to break binding). Because of locality, $P_1$'s behavior is independent of what $P_2$ receives (namely $z_2$). Therefore, there is only \emph{one} such $z_2$ which $V_0$ will ultimately accept as a valid unveiling of $c$ in both ways (recall that our commitment is statistically binding).

Therefore, in the worst case, for every commitment there exists a sequence of interactions between $V_1$ and $P_1$ such that $P_1$ will attempt to break the binding of that commitment. Each such commitment-breaking corresponds to at most one string $z_2$ that will actually work.

Let us denote the set of such binding-breaking strings by $B$. If $z_2 \notin B$, then the provers \emph{will not break binding}, and the soundness error is reduced to that of the underlying protocol (at most $1/3$). On the other hand, since $\left| B \right| < \mathbf{poly}(k)$, the probability that $z_2 \in B$ is at most $ \mathbf{poly}(k)/2^k$.

Therefore, the soundness error of our protocol is at most $$Pr[ z_2 \notin B\text { and underlying protocol accepts}] + Pr[z_2 \in B] \leq \frac{1}{3} + \frac{\mathbf{poly}(k)}{2^k}.$$

\subsubsection{Zero-Knowledge}

The simulation will be divided in two parts. In the first part, the simulator produces a transcript of the \emph{pre-computation}, \emph{multilinearity test} and \emph{sumcheck with oracle} parts, which involves only interactions with $V_1$. In the second part, the simulator will fake a valid \emph{consistency test}.

\begin{quote}
\rule{\linewidth}{1pt}
\begin{protocol}\label{CCSIM1} ( Perfectly Indistinguishable, $\mathbf{PR}$-Local Simulator for Protocol \ref{IsolatingZKMIP}, Part 1) \end{protocol}

The setup:

\begin{itemize}

\item Let $( \widehat{S}, S_1, S_2)$ be a set of locality-explicit simulators.

\item $S_1$ and $S_2$ can send $\widehat{S}$ an index along with a bit.

\item $\widehat{S}$ completes the indexed $\mathbf{PR}$ box (protocol \ref{MIP_CHSH}) for both simulators.

\end{itemize}

The simulation strategy:

\begin{enumerate}

\item The simulators agree on unique indices for every commitment used in the protocol.

\item $S_1$ interacts with $V_1$ the way $P_1$ would. Whenever $P_1$ should commit, $S_1$ commits to random bits, just like the single-simulator from section \ref{SEC:NEXP}.

\item For each commitment, $V_2$ sends $S_2$ a string $s$. $S_2$ sends to $\widehat{S}$ the index of the commitment and $s$.

\item $\widehat{S}$ runs the $\mathbf{PR}$ box (protocol \ref{MIP_CHSH}) and replies with $V_2\!$'s half of the output.

\item Whenever $S_1$ needs to unveil a commitment, it can be unveiled in the way $S_1$ desires by sending the corresponding index and bit to $\widehat{S}$.

\item $\widehat{S}$ completes the corresponding $\mathbf{PR}$ box which outputs $t$. $\widehat{S}$ sends $t$ to $S_1$.

\item $S_1$ sends $t$ to $V_1$.

\end{enumerate}

\rule{\linewidth}{1pt}
\end{quote}

The second part (the consistency test) can be done by having the simulators ignore the question.

\begin{quote}
\rule{\linewidth}{1pt}
\begin{protocol}\label{CCSIM2} ( Perfectly Indistinguishable, $\mathbf{PR}$-Local Simulator for Protocol \ref{IsolatingZKMIP}, Part 2) \end{protocol}

\begin{enumerate}

\item $V_1$ sends $i$ to $S_1$.

\item $S_1$ computes $\text{\fbox{$\Omega_1$}} = H_{\text{\fbox{$\gamma$}}}\!\left( Q_{i}  \right)$.

\item Using $\widehat{S}$ to break binding, $S_1$ convinces $V_1$ that $\text{\fbox{$\Omega_1$}}$ is actually $\text{\fbox{$A(Q_i)$}} \oplus H_{\text{\fbox{$\gamma$}}}\!\left( Q_{i}  \right)$.

\item $S_1$ unveils $\text{\fbox{$\Omega_1$}}$ for $V_1$, who gets $\Omega_1 = H_{\text{{$\gamma$}}}\!\left( Q_{i}  \right)$.

\item $V_2$ sends $Q_i'$ to $S_2$.

\item $S_2$ responds with $\Omega_2 = H_{\text{{$\gamma$}}}\!\left( Q_{i}'  \right)$.

\end{enumerate}

\rule{\linewidth}{1pt}
\end{quote}

By the properties of the strongly-universal-2 hash $H$, if $Q_i = Q_i'$ then $\Omega_1 = \Omega_2$. Otherwise $\Omega_1 \neq \Omega_2$ with probability exponentially close to one. This produces the result as desired. The simulators then feed the transcripts to $V_0$, and terminates simulation.

\section{Babai, Fortnow and Lund's MIP for Languages in NEXP}\label{SEC:APP_BFL}

This section describes a variant of the multi-prover protocol for oracle-3-SAT found in \cite{BFL90}. We refer to this as the BFL protocol, or BFL classic.
\begin{definition}

Let $r, s > 0$ be integers. Let $z, b_1, b_2, b_3$ be strings of variables, where $|z| = r$ and $|b_i| = s$. Let $B(z, b_1, b_2, b_3, t_1, t_2, t_3)$ be a Boolean formula in $r+3s+3$ variables. A Boolean function $A : \{0, 1\}^s \rightarrow \{0, 1\}$ is a \textit{3-satisfying oracle} for $B$ if $$B(z, b_1, b_2, b_3, A(b_1), A(b_2), A(b_3)) = 1$$ for every string $z, b_1, b_2, b_3$.

$B$ is \textit{oracle-3-satisfiable} if such a function $A$ exists.

The \textit{Oracle-3-SAT} problem $(B, r, s)$ asks whether a Boolean formula $B$ is oracle-3-satisfiable, where $r$ and $s$ denote the lengths of $z$ and $b_i$, as above.

\end{definition}

\begin{lemma}

Oracle-3-SAT is $\mathbf{NEXP}$-complete.

\end{lemma}

\begin{definition}
Let $\mathbb{F}$ be an arbitrary field.
Let $\phi : \{0, 1\}^m \rightarrow \{0, 1 \}$ be a Boolean function. An \emph{arithmetization} of $\phi$ is a polynomial $f(x_1,\ldots,x_m) \in \mathbb{F}[X_1, \ldots, X_m]$ such that for all $z \in \{0, 1\}^m$, $\phi(z) = 0 \Leftrightarrow f(z) = 0$. A specific one is given in \cite{BFL90}, proposition 3.1 .\\

Equivalently, the $\phi(z) = 0 \Leftrightarrow f(z) = 0$ condition can be replaced with $\phi(z) = 1 \Leftrightarrow f(z) = 0$.

\end{definition}

\begin{quote}
\rule{\linewidth}{1pt}
\begin{protocol}\label{PROT:BFL_sumcheck} ( Sumcheck Protocol ) \end{protocol}

Let $\phi(x_1, \ldots , x_m)$ be the 3-CNF formula which the prover $P$ is trying to show to be a tautology to a verifier $V$.
Let $\mathbb{F}$ be a field of sufficient size (of order at least $(3c+1)m$ will suffice where $c$ is the number of clauses of $\phi$).

\begin{enumerate}

\item $V$ takes $\phi$ and computes its arithmetization $f$ according to \cite{BFL90} Proposition 3.1 and sends it to $P$.

\item $V$ and $P$ agree on a set $I \subset \mathbb{F}$ of size at least $2dm$ where $d$ is the degree of $f$.

\item $V$ assigns $b_0 = 0$, which is supposed to be equal to the sum $$\sum_{x_1 = 0}^1 \ldots \sum_{x_m = 0}^1 f(x_1, \ldots , x_m)^2 = 0$$

\item $i \leftarrow 1$.

\item $P$ sends the coefficients of the univariate polynomial in $x$, $$g_i(x) = h(r_1, \ldots , r_{i-1}, x) = \sum_{x_{i+1}=0}^1 \ldots \sum_{x_{m}=0}^1 f(r_1, \ldots, r_{i-1}, x, x_{i+1}, \ldots, x_m)^2$$

\item $V$ checks whether $b_{i-1} = g_i(0) + g_i(1)$. If not, abort.

\item $V$ chooses a random $r_i \in I$, computes $b_i = g_i(r_i)$ and sends $r_i$ to $P$.

\item If $i \leq m$ then $i \leftarrow i+1$ and go to step 4.

\item $V$ checks whether $b_m = f(r_1, \ldots , r_m)^2$.

\end{enumerate}

\rule{\linewidth}{1pt}
\end{quote}


%
\begin{quote}
\rule{\linewidth}{1pt}
\begin{protocol}\label{PROT:BFL_appendix} ( Babai, Fortnow and Lund's MIP for Oracle-3-SAT ) \end{protocol}

Given $(B, r, s)$ as common input.

\begin{enumerate}


\item (sumcheck with oracle) $V$ and $P_1$ execute protocol \ref{PROT:BFL_sumcheck}. Let $( Q_{k+1},Q_{k+2},Q_{k+3} ) = ( r_{r+1}...r_{r+s},r_{r+s+1}...r_{r+2s},r_{r+2s+1}...r_{r+3s} ) \in (I^{s})^{3}$ be $V$'s questions during this phase.

\item (multilinearity test) $V$ asks $P_1$ to simulate an oracle storing the function $A$. $V$ queries $P_1$ with random, linearly related values in $I^{s}$.
If any response does not satisfy linearity, abort protocol. Let $Q_1,\ldots,Q_k \in I^{s}$ be $V$'s questions during this phase.

\item (non-adaptiveness test) $V$ chooses uniformly at random an $i$ such that $1 \leq i \leq k+3$ and asks $Q_i$ to $P_2$. If $P_2$'s answer differs from that of $P_1$, reject. Otherwise accept.

\end{enumerate}

\rule{\linewidth}{1pt}
\end{quote}

\section{Non-Locality -- an introduction}\label{AppB}
In this section we solely focus on the two-party single-round games and strategies that are sufficient
to discuss and analyze most of the MIPs.
Definitions and proofs for complete generalizations to multi-party multi-round games and strategies
will appear in a forthcoming paper with co-author Adel Magra.

\subsubsection{Games:}

Let $V$ be a predicate on $A\times B\times X\times Y$ (for some finite sets $A, B, X,$ and $Y$) and let $\pi$ be a probability distribution on $A\times B$.
Then $V$ and $\pi$ define a (single-round) game $G$ as follows: A pair of questions $(a,b)$ is randomly chosen according to distribution $\pi$, and $a\in A$ is sent to Alice
and $b\in B$ is sent to Bob. Alice must respond with an answer $x\in X$ and Bob with an answer $y\in Y$.
Alice and Bob win if $V$ evaluates to 1 on $(a,b,x,y)$ and lose otherwise.

\subsubsection{Strategies: Two-Party Channels}

A strategy for Alice and Bob is simply a probability distribution $P_{(x,y | a,b)}$ describing exactly how they will answer $(x,y)$ on every
pair of questions $(a,b)$. We now breakdown the set of all possible strategies for Alice and Bob according to their {\em non-locality}.

\subsubsection{Deterministic and Local Strategies:} A strategy $P_{(x,y | a,b)}$ is {\em deterministic} if there exists functions $f_{A}:A\rightarrow X, f_{B}:B\rightarrow Y$ such that
$$P_{(x,y | a,b)} = 
\begin{cases}
    1  & \text{if $x=f_{A}(a)$ and $y=f_{B}(b)$} \\
    0  & \text{otherwise}
\end{cases}.$$
A deterministic strategy corresponds to the situation where Alice and Bob agree on their individual actions before any knowledge of the values $a,b$ is provided to them.
In this case they use only their own input to determine their individual output.


A strategy $P_{(x,y | a,b)}$ is {\em local} if there exists a finite set $R$ and functions $f_{A}:A\times R\rightarrow X, f_{B}:B\times R\rightarrow Y$ such that
$$P_{(x,y | a,b)} = \frac{ |\{ r\in R : x=f_{A}(a,r) \text{ and } y=f_{B}(b,r) | }{ |R|}.$$
A local strategy corresponds to the situation where Alice and Bob agree on a deterministic strategy selected uniformly among $|R|$ such possibilities.
The choice $r$ of Alice and Bob's strategy, and the choice of inputs $(a,b)$ provided to Alice and Bob are generally agreed to be statistically independent random variables.

\subsection{Local Reducibility}
We now turn to the notion of locally reducing a strategy to another, that is how Alice and Bob limited to local strategies but equipped
with a particular (not necessarily local) strategy $U'$ are able to achieve another particular (not necessarily local) strategy $U$. For this purpose
we introduce a notion of distance between strategies in order to analyze strategies that are approaching each other asymptotically.

\subsubsection{Distances between Strategies:} Several distances could be selected here as long as their meaning as it approaches zero are the same. In the definitions below,
$U,U'$ are strategies and ${\mathcal U'}$ is a finite set of strategies. 

\begin{definition}
$ |U,U'| = \displaystyle{ \sum_{a,b,x,y} } | P_{U}(x,y | a,b) - P_{U'}(x,y | a,b) | $
\end{definition}

\begin{definition}
$ |U,{\mathcal U'}| =\displaystyle{ \min_{U'\in {\mathcal U'} } } |U,U'| $
\end{definition}

\subsubsection{Local extensions of Strategies:} For natural integer $n$, we define the set $\operatorname{LOC}^{n}(U)$ of strategies that are local extensions (of order $n$)
of $U$ to be all the strategies Alice and Bob can achieve using local strategies where strategy $U$ may be used up to $n$ times as sub-routine
calls\footnote{Done by selecting functions $f_{A}^{0}:A\times R\rightarrow A, ~f_{A}^{1}:A\times X \times R\rightarrow A,..., ~f_{A}^{n-1}:A \times X^{n-1}\times R\rightarrow A$, $ ~f_{A}^{n}:A \times X^{n}\times R\rightarrow X$ to determine the input of each sub-routine from input $a$ and previous outputs.}. If we restrict all the functions used to be polynomial-time computable we analogously define $\underset{\text{\small{poly}}}{\operatorname{LOC}}^{n}(U)$.

\begin{definition}
$U'$ Locally (poly-)Reduces to $U$ ($U' \leq_{\underset{\text{\tiny{(poly)}}}{\mbox{\scriptsize\em LOC}}} U\!$) iff $\displaystyle{\lim_{n\rightarrow \infty}} | U' , \underset{\text{\small{(poly)}}}{\mbox{\em LOC}}^{n}(U) | =0$.
\end{definition}

\begin{definition}
$U'$ is Locally (poly-)Equivalent to $U$ ($U' =_{\underset{\text{\tiny{(poly)}}}{\mbox{\scriptsize\em LOC}}} U\!$) iff
$U' \leq_{\underset{\text{\tiny{(poly)}}}{\mbox{\scriptsize\em LOC}}} U \leq_{\underset{\text{\tiny{(poly)}}}{\mbox{\scriptsize\em LOC}}} U'.$
\end{definition}

\subsubsection{Non-Adaptive extensions of Strategies:} For natural integer $n$, we define the set $\operatorname{NAD}^{n}(U)$ of strategies that are Non-Adaptive extensions (of order $n$)
of $U$ to be all the strategies Alice and Bob can achieve using Non-Adaptive strategies where strategy $U$ may be used up to $n$ times as sub-routine
calls\footnote{Done by selecting functions $f_{A}^{0}:A\times R\rightarrow A, ~f_{A}^{1}:A \times R\rightarrow A,..., ~f_{A}^{n-1}:A \times R\rightarrow A$,\\ $ ~f_{A}^{n}:A \times X^{n}\times R\rightarrow X$ to determine the input of each sub-routine from input $a$ only.}. If we restrict the functions used to be poly-time computable we get $\underset{\text{\small{poly}}}{\operatorname{NAD}}^{n}(U)$.

\begin{definition}
$U'$ Non-Adaptively (poly-)Reduces to $U$ ($U' \leq_{\underset{\text{\tiny{(poly)}}}{\mbox{\scriptsize\em NAD}}} U\!$) iff $\displaystyle{\lim_{n\rightarrow \infty}} | U' , \underset{\text{\small{(poly)}}}{\mbox{\em NAD}}^{n}(U) | =0$.
\end{definition}

\begin{definition}
$U'$ is Non-Adaptively (poly-)Equivalent to $U$ ($U' =_{\underset{\text{\tiny{(poly)}}}{\mbox{\scriptsize\em NAD}}} U\!$) iff 
$U' \leq_{\underset{\text{\tiny{(poly)}}}{\mbox{\scriptsize\em NAD}}} U \leq_{\underset{\text{\tiny{(poly)}}}{\mbox{\scriptsize\em NAD}}} U'.$
\end{definition}

In general, Non-Adaptive reducibility is a weaker notion than local reducibility. However, for certain distributions $\mathbf{U}$ it may result that $\{ D | D \leq_{\underset{\text{\tiny{(poly)}}}{\mbox{\scriptsize\em LOC}}} \mathbf{U} \}
= \{ D | D \leq_{\underset{\text{\tiny{(poly)}}}{\mbox{\scriptsize\em NAD}}} \mathbf{U} \}$ as follows.

\subsection{Locality}
We now define the lowest of the non-locality classes ${\mathbb{LOC}}$. We could define it directly from the notion of local strategies as defined above, but for analogy with the other classes
we later define, ${\mathbb{LOC}}$ is defined as all those strategies locally reducible to a {\em complete} strategy we call $\mathbf{ID}$ (see {\bf Fig.~\ref{ID}}). Of course, any strategy is complete for this class.
\begin{figure}[h!]

\centering
  \mbox{\Qcircuit @C=1em @R=.7em {
      \lstick{a}  \ar[r] & \multigate{1}{\mathbf{ID}} & \rstick{b} \ar[l] \\
      \lstick{a}   & \ghost{\mathbf{ID}} \ar[r]\ar[l] & \rstick{b}
    }}
\caption{an $\mathbf{ID}$-box}
  \label{ID}
\end{figure}

\begin{definition}
${\mathbb{LOC}} = \{ U | U \leq_{\mbox{\scriptsize\em LOC}} \mathbf{ID} \}$ and $\underset{{{poly}}}{\mathbb{LOC}} = \{ U | U \leq_{\underset{\text{\tiny{poly}}}{\mbox{\scriptsize\em LOC}}} \mathbf{ID} \}$
\end{definition}

Note: ${\mathbb{LOC}}$ is the class of strategies that John Bell \cite{BELL64} considered as classical hidden-variable theories that he compared to entanglement.
It is also the class of strategies that BenOr, Goldwasser, Kilian and Wigderson \cite{BGKW88} chose to define classical Provers in Multi-Provers Interactive Proof Systems.
 ${\mathbb{LOC}}$ is also those strategies Non-Adaptively reducible to $\mathbf{ID}$

\begin{definition}
Alternatively, ${\mathbb{LOC}} = \{ U | U \leq_{\mbox{\scriptsize\em NAD}} \mathbf{ID} \}$ and $\underset{{{poly}}}{\mathbb{LOC}} = \{ U | U \leq_{\underset{\text{\tiny{poly}}}{\mbox{\scriptsize\em NAD}}} \mathbf{ID} \}$
\end{definition}

Alternatively, we can also define $\mathbb{LOC}$ from an empty box as used in the core of this paper
\begin{figure}[h!]

\centering
  \mbox{\Qcircuit @C=1em @R=.7em {
      \lstick{a}  \ar[r] & \multigate{1}{\mathbf{\emptyset}} & \rstick{b} \ar[l] \\
      \lstick{x}   & \ghost{\mathbf{\emptyset}} \ar[r]\ar[l] & \rstick{y}
    }}
\caption{an $\mathbf{\emptyset}$-box where $x\in X$ and $y\in Y$ are uniform and independent of everything else}
  \label{PHI}
\end{figure}
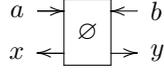

\begin{definition}
Alternatively, ${\mathbb{LOC}} = \{ U | U \leq_{\mbox{\scriptsize\em NAD}} \mathbf{\emptyset} \} = \{ U | U \leq_{\mbox{\scriptsize\em LOC}} \mathbf{\emptyset} \} $
\end{definition}

\subsection{One-Way Signalling}
We now turn to One-Way Signalling which allows communication from one side to the other. We name the directions arbitrarily Left and Right.
We define ${\mathbf R} \text{-} {\mathbb{SIG}}$ (resp. ${\mathbf L} \text{-} {\mathbb{SIG}}$) as all those strategies locally reducible to a {\em complete} strategy we call ${\mathbf R} \text{-} {\mathbf{SIG}}$ (see {\bf Fig.~\ref{rSIG}})
(resp. ${\mathbf L} \text{-} {\mathbf{SIG}}$ (see {\bf Fig.~\ref{lSIG}})). These classes are useful to define what it means for a strategy to {\em signal} as well as the notion of {\em No-Signalling} strategies.

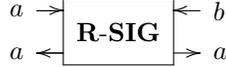
\begin{figure}[h!]

\centering
  \mbox{\Qcircuit @C=1em @R=.7em {
      \lstick{a}  \ar[r] & \multigate{1}{ {\mathbf R} \text{-} {\mathbf{SIG}} } & \rstick{b} \ar[l] \\
      \lstick{a}   & \ghost{ {\mathbf R} \text{-} {\mathbf{SIG}} } \ar[r]\ar[l] & \rstick{a}
    }}
\caption{an $\mathbf{R} \text{-} {\mathbf{SIG}}$-box}
  \label{rSIG}
\end{figure}

\begin{definition}
${\mathbf R} \text{-} {\mathbb{SIG}} = \{ U | U \leq_{\mbox{\scriptsize\em LOC}} {\mathbf R} \text{-} {\mathbf{SIG}} \}$ and ${\mathbf R} \text{-} \underset{{{poly}}}{\mathbb{SIG}} = \{ U | U \leq_{\underset{\text{\tiny{poly}}}{\mbox{\scriptsize\em LOC}}} {\mathbf R} \text{-} {\mathbf{SIG}} \}$
\end{definition}

\begin{definition}
We say that $U$ Right Signals (is ${\mathbf R} \text{-} {\mathbb{SIG}}$-verbose\footnote{We define the notion of $\mathbb{L}$-verbose in analogy to $\mathbb{NP}$-hard: it means ``as verbose as any distribution in non-locality class
$\mathbb{L}$''. In consequence, a distribution $U$ is $\mathbb{L}$-complete if $U \in \mathbb{L}$ and $U$ is $\mathbb{L}$-verbose.}) iff ${\mathbf R} \text{-} {\mathbf{SIG}} \leq_{\mbox{\scriptsize\em LOC}} U$.
\end{definition}

\begin{figure}[h!]

\centering
  \mbox{\Qcircuit @C=1em @R=.7em {
      \lstick{a}  \ar[r] & \multigate{1}{{\mathbf L} \text{-} {\mathbf{SIG}}} & \rstick{b} \ar[l] \\
      \lstick{b}   & \ghost{{\mathbf L} \text{-} {\mathbf{SIG}}} \ar[r]\ar[l] & \rstick{b}
    }}
\caption{an ${\mathbf L} \text{-} {\mathbf{SIG}}$-box}
  \label{lSIG}
\end{figure}
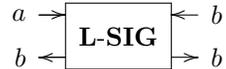

\begin{definition}
${\mathbf L} \text{-} {\mathbb{SIG}} = \{ U | U \leq_{\mbox{\scriptsize\em LOC}} {\mathbf L} \text{-} {\mathbf{SIG}} \}$ and ${\mathbf L} \text{-} \underset{{{poly}}}{\mathbb{SIG}} = \{ U | U \leq_{\underset{\text{\tiny{poly}}}{\mbox{\scriptsize\em LOC}}} {\mathbf L} \text{-} {\mathbf{SIG}} \}$

\end{definition}

\begin{definition}
We say that $U$ Left Signals (is ${\mathbf L} \text{-} {\mathbb{SIG}}$-verbose) iff ${\mathbf L} \text{-} {\mathbf{SIG}} \leq_{\mbox{\scriptsize\em LOC}} U$.
\end{definition}

\begin{definition}
We say that $U$ Signals iff $U$ Right Signals or Left Signals.
\end{definition}

We prove a first result that is intuitively obvious. We show that the complete strategy ${\mathbf R} \text{-} {\mathbf{SIG}}$ cannot be approximated in ${\mathbf L} \text{-} {\mathbb{SIG}}$
and the other way around.

\begin{theorem}\label{LR-impossible}
${\mathbf R} \text{-} {\mathbf{SIG}} \not\in {\mathbf L} \text{-} {\mathbb{SIG}}$ and
${\mathbf L} \text{-} {\mathbf{SIG}} \not\in {\mathbf R} \text{-} {\mathbb{SIG}}$.
\end{theorem}

\begin{proof}
Follows from a simple capacity argument. For all $n$, all the channels in $\mbox{LOC}^{n}({\mathbf R} \text{-} {\mathbf{SIG}})$ have zero left-capacity, while ${\mathbf L} \text{-} {\mathbf{SIG}}$ has non-zero left-capacity. And vice-versa.
\end{proof}

\subsection{Signalling}
We are now ready to define the largest of the non-locality classes ${\mathbb{SIG}}$. Indeed every possible strategy is in ${\mathbb{SIG}}$.

\begin{definition}

${\mathbb{SIG}} = \{ U | U \leq_{\mbox{\scriptsize\em LOC}} \mathbf{SIG}\}$ and $\underset{{{poly}}}{\mathbb{SIG}} = \{ U | U \leq_{\underset{\text{\tiny{poly}}}{\mbox{\scriptsize\em LOC}}} {\mathbf{SIG}} \}$

\end{definition}

\begin{figure}[h!]

\centering
  \mbox{\Qcircuit @C=1em @R=.7em {
      \lstick{a}  \ar[r] & \multigate{1}{{\mathbf{SIG}}} & \rstick{b} \ar[l] \\
      \lstick{b}   & \ghost{{\mathbf{SIG}}} \ar[r]\ar[l] & \rstick{a}
    }}
\caption{a ${\mathbf{SIG}}$-box}
  \label{SIG}
\end{figure}
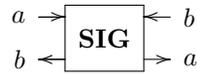

\begin{definition}
We say that $U$ Fully Signals (is ${\mathbb{SIG}}$-verbose) iff $U$ Right Signals and Left Signals.
\end{definition}

\begin{figure}[hbt]
\begin{center}
\fbox{
\includegraphics[width=1.0\textwidth]{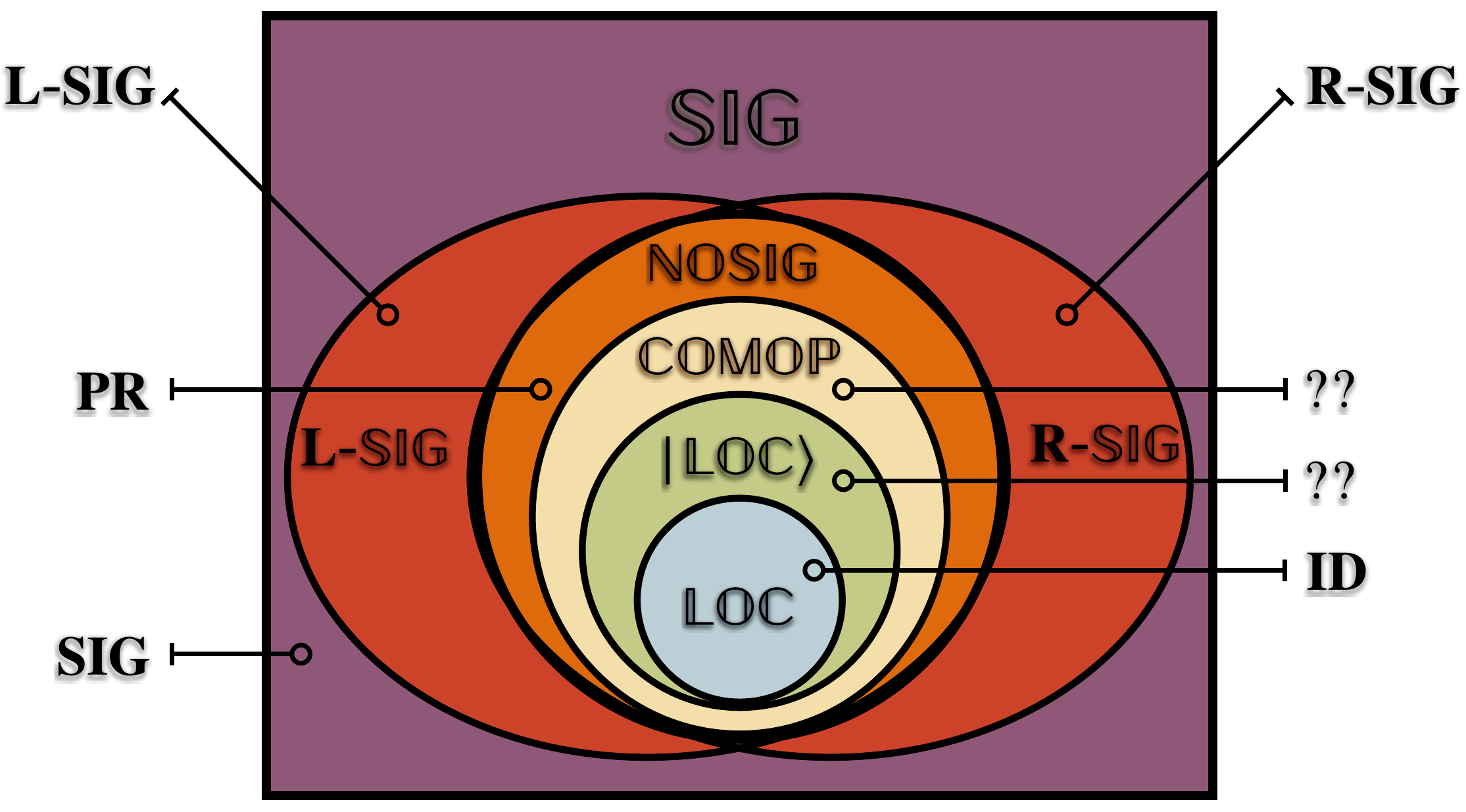}
}
\end{center}
\caption{\label{HIER} Non-locality Hierarchy and complete (two-party) distributions in each class.}
\end{figure}

\subsection{No-Signalling}
We finally define the less intuitive non-locality class ${\mathbb{NOSIG}}$ in relation to classes defined above. 

\begin{definition}
$\mathbb{NOSIG} = {\mathbf R} \text{-} {\mathbb{SIG}} \bigcap {\mathbf L} \text{-} {\mathbb{SIG}}$ and $\underset{{{poly}}}{\mathbb{NOSIG}} = {\mathbf R} \text{-} \underset{{{poly}}}{\mathbb{SIG}} \bigcap {\mathbf L} \text{-} \underset{{{poly}}}{\mathbb{SIG}}$.
\end{definition}

A similar characterization may be found in \cite{Acin2015} Section 3 and \cite{Barnum05} Corollary 3.5.

\begin{theorem} \label{NSig}.
The above definition of $\mathbb{NOSIG}$ exactly coincides with the {\em traditional} notion of No-Signalling \cite{BLM+05}.
\end{theorem}

Intuitively, a distribution $P(x,y | a,b)$ is No-Signalling as long as for every $a$ the $x|b$ and for every $b$ the $y|a$  channels have zero capacity.
 
Note: Forster and Wolf \cite{PhysRevA.84.042112} have proved that {\bf PR} (see {\bf Fig.~\ref{nlbox}}) is complete for $\mathbb{NOSIG}$ distributions under an asymptotic definition similar to ours.

\begin{figure}[h!]
\centering
  \mbox{\Qcircuit @C=1em @R=.7em {
      \lstick{a}  \ar[r] & \multigate{1}{{\bf PR}} & \rstick{b} \ar[l] \\
      \lstick{x}   & \ghost{{\bf PR}} \ar[r]\ar[l] & \rstick{y}
    }}
\caption{a {\bf PR}-box satisfying the CHSH condition,\newline that $a \wedge b = x \oplus y$, uniformly among solutions}
  \label{nlbox}
\end{figure}
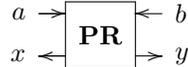

{\bf Fig.~\ref{HIER}} shows the relation of these classes as well as the case obtained via quantum entanglement (${\mathbb{|LOC\rangle }}$) as considered by Bell \cite{BELL64} and via commuting-operators (${\mathbb{COMOP}}$) as defined by 
Ito, 
Kobayashi, 
Preda, 
Sun, and 
Yao \cite{DBLP:conf/coco/ItoKPSY08}. We include those for completeness but will not discuss these particular classes any further in this work.

\begin{definition}
We say that $U$ does not Signal iff $U$ does not Right Signal nor Left Signal iff $U \in \mathbb{NOSIG}$.
\end{definition}

\section{Visual description of the new model}

\subsection{Local Multi-Prover Interactive Proofs}

%
%

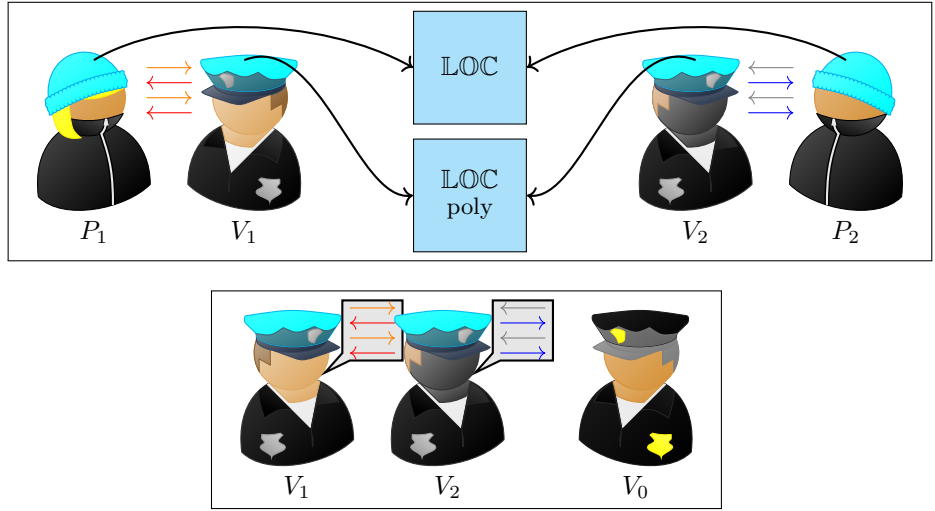
\begin{figure}[htb]
\begin{center}
\fbox{
\begin{tikzpicture}
\node[draw,rectangle,color=black,fill=cyan!30,minimum size=1.5cm] (U) at (5,-0.7) { $\underset{\text{\small{poly}}}{\mathbb {LOC}}$ };
\node[draw,rectangle,color=black,fill=cyan!30,minimum size=1.5cm] (V) at (5,1) { ${\mathbb {LOC}}$ };
\node[criminal,saturated,female,shirt=black,hat=cyan,hair=yellow,minimum size=1.5cm] (A) at (0,0.1) {$P_{1}$};
\node[police,mirrored,shirt=black,hat=cyan,hatbadge=gray,badge=gray,minimum size=1.5cm] (B) at (2,0.1) {$V_{1}$};
\node[police,shirt=black,hat=cyan,skin=black,hair=brown,hatbadge=gray,badge=gray,minimum size=1.5cm] (C) at (8,0.1) {$V_{2}$};
\node[criminal,saturated,mirrored,shirt=black,hat=cyan,minimum size=1.5cm] (D) at (10,0.1) {$P_{2}$};
\draw[orange, ->] (0.7,1) -- (1.3,1);
\draw[red, <-] (0.7,0.8) -- (1.3,0.8);
\draw[orange, ->] (0.7,0.6) -- (1.3,0.6);
\draw[red, <-] (0.7,0.4) -- (1.3,0.4);
\draw[gray, ->] (9.3,1) -- (8.7,1);
\draw[blue, <-] (9.3,0.8) -- (8.7,0.8);
\draw[gray, ->] (9.3,0.6) -- (8.7,0.6);
\draw[blue, <-] (9.3,0.4) -- (8.7,0.4);
\draw[black, thick, ->] (A.north) .. controls (1,1.7)  and (3,1.7)  .. (V.west);
\draw[black, thick, ->] (D.north) .. controls (9,1.7)  and (7,1.7)  .. (V.east);
\draw[black, thick, ->] (B.north)  .. controls (3,1.5) and (3.5,-0.8) .. (U.west);
\draw[black, thick, ->] (C.north)  .. controls (7,1.5) and (6.5,-0.8) ..  (U.east);
\end{tikzpicture}
}
\end{center}
%
\begin{center}
\fbox{
\begin{tikzpicture}
\node[draw,rectangle,color=white,minimum size=1.5cm] at (5,-0.7) {};
\node[police,shirt=black,hat=cyan,badge=gray,hatbadge=gray,minimum size=1.5cm] (A) at (0.5,0) {$V_{1}$};
\draw[black,fill=black!10, thick] (1.1,0.5) -- (1.1,1.2) -- (1.9,1.2) -- (1.9,0.4) -- (1.2,0.4) -- (A.mouth) -- cycle;
\draw[orange, ->] (1.2,1.1) -- (1.8,1.1);
\draw[red, <-] (1.2,0.9) -- (1.8,0.9);
\draw[orange, ->] (1.2,0.7) -- (1.8,0.7);
\draw[red, <-] (1.2,0.5) -- (1.8,0.5);
\node[police,saturated,shirt=black,hat=black,hatshield=gray,mirrored,hair=gray,hatbadge=yellow,badge=yellow,minimum size=1.5cm] (B) at (5,0) {$V_{0}$};
\node[police,shirt=black,hat=cyan,skin=black,hair=brown,hatbadge=gray,badge=gray,minimum size=1.5cm] (C) at (2.5,0) {$V_{2}$};
\draw[black,fill=black!10, thick] (3.1,0.5) -- (3.1,1.2) -- (3.9,1.2) -- (3.9,0.4) -- (3.2,0.4) -- (C.mouth) -- cycle;
\draw[gray, ->] (3.8,1.1) -- (3.2,1.1);
\draw[blue, <-] (3.8,0.9) -- (3.2,0.9);
\draw[gray, ->] (3.8,0.7) -- (3.2,0.7);
\draw[blue, <-] (3.8,0.5) -- (3.2,0.5);
\end{tikzpicture}
}
\end{center}
\caption{\label{inter} Interrogation phase (top) followed by decision phase (bottom).}
\end{figure}

In the Interrogation phase (see {\bf Fig. \ref{inter}}) $V_{1},...,V_{k}$ (equipped with an arbitrary local correlator) individually interrogate $P_{1},...,P_{k}$ (equipped with an arbitrary local correlator).
At the end of the interactive part, all the $V_{1},...,V_{k}$ report to $V_{0}$ who takes the final decision. The corresponding complexity class is $\mathbf{MIP} = \mathbf{MIP}^{{\mathbb {LOC}}}_{\underset{\text{\tiny{poly}}}{\mathbb {LOC}}} = \mathbf{NEXP}$.

%

\subsection{Entangled Multi-Prover Interactive Proofs}
\label{ent-def}

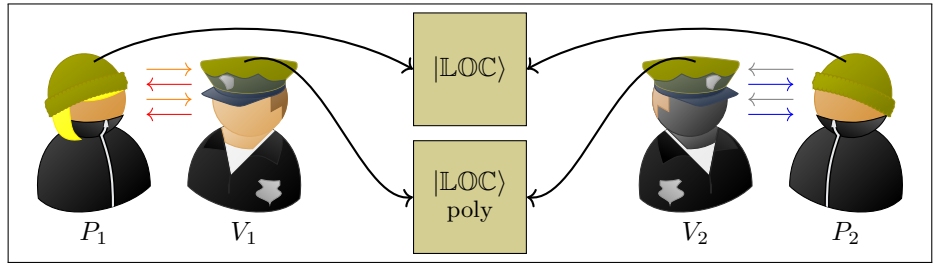
\begin{figure}[htb]
\begin{center}
\fbox{
\begin{tikzpicture}
\node[draw,rectangle,color=black,fill=olive!40,minimum size=1.5cm] (U) at (5,-0.7) {$\underset{\text{\small{poly}}}{|{\mathbb {LOC}}\rangle}$};
\node[draw,rectangle,color=black,fill=olive!40,minimum size=1.5cm] (V) at (5,1) {$|{\mathbb {LOC}}\rangle$};
\node[criminal,saturated,female,shirt=black,hat=olive,hair=yellow,minimum size=1.5cm] (A) at (0,0.1) {$P_{1}$};
\node[police,mirrored,shirt=black,hat=olive,hatbadge=gray,badge=gray,minimum size=1.5cm] (B) at (2,0.1) {$V_{1}$};
\node[police,shirt=black,hat=olive,skin=black,hair=brown,hatbadge=gray,badge=gray,minimum size=1.5cm] (C) at (8,0.1) {$V_{2}$};
\node[criminal,saturated,mirrored,shirt=black,hat=olive,minimum size=1.5cm] (D) at (10,0.1) {$P_{2}$};
\draw[orange, ->] (0.7,1) -- (1.3,1);
\draw[red, <-] (0.7,0.8) -- (1.3,0.8);
\draw[orange, ->] (0.7,0.6) -- (1.3,0.6);
\draw[red, <-] (0.7,0.4) -- (1.3,0.4);
\draw[gray, ->] (9.3,1) -- (8.7,1);
\draw[blue, <-] (9.3,0.8) -- (8.7,0.8);
\draw[gray, ->] (9.3,0.6) -- (8.7,0.6);
\draw[blue, <-] (9.3,0.4) -- (8.7,0.4);
\draw[black, thick, ->] (A.north) .. controls (1,1.7)  and (3,1.7)  .. (V.west);
\draw[black, thick, ->] (D.north) .. controls (9,1.7)  and (7,1.7)  .. (V.east);
\draw[black, thick, ->] (B.north)  .. controls (3,1.5) and (3.5,-0.8) .. (U.west);
\draw[black, thick, ->] (C.north)  .. controls (7,1.5) and (6.5,-0.8) ..  (U.east);
\end{tikzpicture}
}
\end{center}
\caption{\label{interQ} Interrogation phase.}
\end{figure}

In the Interrogation phase (see {\bf Fig. \ref{interQ}}) $V_{1},...,V_{k}$ (equipped with an arbitrary entangled correlator) individually interrogate $P_{1},...,P_{k}$ (equipped with an arbitrary entangled correlator).
At the end of the interactive part, all the $V_{1},...,V_{k}$ report to $V_{0}$ who takes the final decision.
The corresponding complexity class is $\mathbf{MIP}^{*} = \mathbf{MIP}^{{\mathbb{| LOC \rangle }}}_{\underset{\text{\tiny{poly}}}{\mathbb {| LOC \rangle }}} \supseteq \mathbf{NEXP}$.

\subsection{No-Signalling Multi-Prover Interactive Proofs}
\label{NS-def}

%

\begin{figure}[htb]
\begin{center}
\fbox{
\begin{tikzpicture}
\node[draw,rectangle,color=black,fill=orange!90,minimum size=1.5cm] (U) at (5,-0.7) {$\underset{\text{\small{poly}}}{\mathbb {NOSIG}}$};
\node[draw,rectangle,color=black,fill=orange!90,minimum size=1.5cm] (V) at (5,1) {${\mathbb {NOSIG}}$};
\node[criminal,saturated,female,shirt=black,hat=orange,hair=yellow,minimum size=1.5cm] (A) at (0,0.1) {$P_{1}$};
\node[police,mirrored,shirt=black,hat=orange,hatbadge=gray,badge=gray,minimum size=1.5cm] (B) at (2,0.1) {$V_{1}$};
\node[police,shirt=black,hat=orange,skin=black,hair=brown,hatbadge=gray,badge=gray,minimum size=1.5cm] (C) at (8,0.1) {$V_{2}$};
\node[criminal,saturated,mirrored,shirt=black,hat=orange,minimum size=1.5cm] (D) at (10,0.1) {$P_{2}$};
\draw[orange, ->] (0.7,1) -- (1.3,1);
\draw[red, <-] (0.7,0.8) -- (1.3,0.8);
\draw[orange, ->] (0.7,0.6) -- (1.3,0.6);
\draw[red, <-] (0.7,0.4) -- (1.3,0.4);
\draw[gray, ->] (9.3,1) -- (8.7,1);
\draw[blue, <-] (9.3,0.8) -- (8.7,0.8);
\draw[gray, ->] (9.3,0.6) -- (8.7,0.6);
\draw[blue, <-] (9.3,0.4) -- (8.7,0.4);
\draw[black, thick, <->] (A.north) .. controls (1,1.7)  and (3,1.7)  .. (V.west);
\draw[black, thick, <->] (D.north) .. controls (9,1.7)  and (7,1.7)  .. (V.east);
\draw[black, thick, <->] (B.north)  .. controls (3,1.5) and (3.5,-0.8) .. (U.west);
\draw[black, thick, <->] (C.north)  .. controls (7,1.5) and (6.5,-0.8) ..  (U.east);
\end{tikzpicture}
}
\end{center}
\caption{\label{NoSIG} Interrogation phase.}
\end{figure}

In the Interrogation phase (see {\bf Fig. \ref{NoSIG}}) $V_{1},...,V_{k}$ (equipped with an arbitrary No-Signalling correlator) individually interrogate $P_{1},...,P_{k}$ (equipped with an arbitrary No-Signalling correlator).
At the end of the interactive part, all the $V_{1},...,V_{k}$ report to $V_{0}$ who takes the final decision.
The corresponding complexity class is $\mathbf{MIP}^{ns} = \mathbf{MIP}^{{\mathbb {NOSIG}}}_{\underset{\text{\tiny{poly}}}{\mathbb {NOSIG}}}  = \mathbf{EXP}$.

As noted before, most MIPs found in the literature are actually (non-adaptive) local-verifier MIPs (see {\bf Fig. \ref{NoSIG-LOCAL}}) yielding for instance
$\mathbf{MIP}^{ns} = \mathbf{MIP}^{{\mathbb {NOSIG}}}_{\underset{\text{\tiny{poly}}}{\mathbb {LOC}}}$.

\begin{figure}[htb]
\begin{center}
\fbox{
\begin{tikzpicture}
\node[draw,rectangle,color=black,fill=cyan!30,text=black,minimum size=1.5cm] (U) at (5,-0.7) {$\underset{\text{\small{poly}}}{\mathbb {LOC}}$};
\node[draw,rectangle,color=black,fill=orange,minimum size=1.5cm] (V) at (5,1) {${\mathbb {NOSIG}}$};
\node[criminal,saturated,female,shirt=black,hat=orange,hair=yellow,minimum size=1.5cm] (A) at (0,0.1) {$P_{1}$};
\node[police,mirrored,shirt=black,hat=cyan,hatbadge=gray,badge=gray,minimum size=1.5cm] (B) at (2,0.1) {$V_{1}$};
\node[police,shirt=black,hat=cyan,skin=black,hair=brown,hatbadge=gray,badge=gray,minimum size=1.5cm] (C) at (8,0.1) {$V_{2}$};
\node[criminal,saturated,mirrored,shirt=black,hat=orange,minimum size=1.5cm] (D) at (10,0.1) {$P_{2}$};
\draw[orange, ->] (0.7,1) -- (1.3,1);
\draw[red, <-] (0.7,0.8) -- (1.3,0.8);
\draw[orange, ->] (0.7,0.6) -- (1.3,0.6);
\draw[red, <-] (0.7,0.4) -- (1.3,0.4);
\draw[gray, ->] (9.3,1) -- (8.7,1);
\draw[blue, <-] (9.3,0.8) -- (8.7,0.8);
\draw[gray, ->] (9.3,0.6) -- (8.7,0.6);
\draw[blue, <-] (9.3,0.4) -- (8.7,0.4);
\draw[black, thick, <->] (A.north) .. controls (1,1.7)  and (3,1.7)  .. (V.west);
\draw[black, thick, <->] (D.north) .. controls (9,1.7)  and (7,1.7)  .. (V.east);
\draw[black, thick, ->] (B.north)  .. controls (3,1.5) and (3.5,-0.8) .. (U.west);
\draw[black, thick, ->] (C.north)  .. controls (7,1.5) and (6.5,-0.8) ..  (U.east);
\end{tikzpicture}
}
\end{center}
\caption{\label{NoSIG-LOCAL} Interrogation phase.}
\end{figure}

\subsection{A New, Stronger Flavour of Zero-Knowledge}\label{sec:newzk}
Traditionally zero-knowledge is defined as a property of the honest provers for all (polynomial-time) verifiers
$$\forall_{\text{poly}} V^\prime~ \exists_{\text{poly}} S ~ \forall x\!\in\!L ~ \forall w ~~ {\mathbf{VIEW}}_{V^\prime}[ P_{1},...,P_{k},V^\prime ](w,x) = S(w,x). $$

However, in the present context, the fact that the simulation of $V^\prime$'s view via a single centralized simulator $S$, achieving zero-knowledge is rather easy
because such an $S$ can cheat the binding property of the commitments at will. The intuition behind the original definition is that the verifier is unable to convince a third party
(a Judge $J_{0}$) because the {\bf VIEW} he reports (see {\bf Fig.~\ref{dec}}) could have been equally created (with the same distribution) by a simulator.
Nevertheless, a stronger flavour of zero-knowledge is achieved if the simulator is not invoking
its full signalling power whenever the verifier does not use such power.

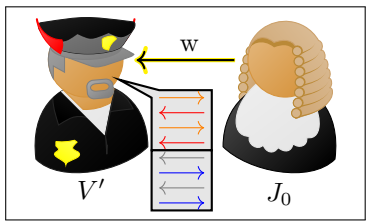
\begin{figure}[htb]
\begin{center}
\fbox{
\begin{tikzpicture}
\node[draw,rectangle,color=white,minimum size=1.5cm] at (5,-0.7) {};
\node[judge,saturated,mirrored,minimum size=1.5cm] (B) at (5,0) {$J_{0}$};
\node[police,evil,saturated,shirt=black,hat=black,hatshield=gray,hair=gray,hatbadge=yellow,badge=yellow,minimum size=1.5cm] (C) at (2.5,0) {$V^\prime$};
\draw[yellow, ultra thick, ->] (4.41, 0.5) -- (3.07,0.5);
\draw[black, thick, ->] (4.4,0.5) -- (3.1,0.5);
\filldraw[black] (3.8,0.5) circle (0pt) node[anchor=south] {w};
\draw[black,fill=black!10, thick] (3.3,0) -- (3.3,-0.7) -- (4.1,-0.7) -- (4.1,0.1) -- (3.4,0.1) -- (C.mouth) -- cycle;
\draw[orange, <-] (4.0,0) -- (3.4,0);
\draw[red, ->] (4.0,-0.2) -- (3.4,-0.2);
\draw[orange, <-] (4.0,-0.4) -- (3.4,-0.4);
\draw[red, ->] (4.0,-0.6) -- (3.4,-0.6);
\draw[black,fill=black!10, thick] (3.3,-0.7) -- (3.3,-1.5) -- (4.1,-1.5) -- (4.1,-0.7) -- cycle;
\draw[gray, ->] (4.0,-0.8) -- (3.4,-0.8);
\draw[blue, <-] (4.0,-1.0) -- (3.4,-1.0);
\draw[gray, ->] (4.0,-1.2) -- (3.4,-1.2);
\draw[blue, <-] (4.0,-1.4) -- (3.4,-1.4);

\end{tikzpicture}
} 
\end{center}
\caption{\label{dec} (Interac/Simula)tion-Distinction phase.}
\end{figure}

For all non-locality levels starting with $\widehat{{\mathbb {S}}}$ and up, the simulators $S_{i}$ do not need more non-local power than the verifiers $V^\prime_{i}$.
The ultimate (strongest) notion of ``$\underset{\text{\small{poly}}}{\mathbb {LOC}}$-local ZK'' being $\mathbf{ZK}^{\overset{\text{\tiny{poly}}}{\mathbb {LOC}}}$ because at all levels $V^\prime$ is simulated by a simulator with no extra
non-local power, whereas at the opposite end of the spectrum $\mathbf{ZK}^{\overset{\text{\tiny{poly}}}{\mathbb {SIG}}}$ is what is generally considered zero-knowledge with a single simulator or a group of signalling simulators.

This stronger notion of zero-knowledge is particularly interesting in the relativistic bit-commitment scenario where a pair of judges
may provide separate auxiliary-inputs to spatially separated verifiers pretending to be speaking to powerful provers. If the verifiers can report
their conversation fast enough to the judges (but not interact with the judges however), they must be able to do so without invoking signalling because of the distance separating them. If a pair of simulators
can produce the same distribution of views in the same context, we obtain a stronger flavour of zero-knowledge (See {\bf Fig.~\ref{inter2}}).

%

%
%


\begin{figure}[htb]
\begin{center}
\fbox{
\begin{tikzpicture}
\node[draw,rectangle,color=white,fill=white,minimum size=1.5cm] (U) at (5,-0.7) {$|{\mathbb {LOC}}\rangle$} ;
\node[draw,rectangle,color=black,fill=gray!40,minimum size=1.5cm] (V) at (5,1) { $\widehat{{\mathbb {V}}^{\prime}}$ };
\node[police,evil,shirt=black,hat=gray,hair=yellow,hatbadge=gray,badge=gray,minimum size=1.5cm] (A) at (0,0.1) {${V^{\prime}_{1}}$};
\node[police,mirrored,evil,shirt=black,hat=gray,skin=black,hair=brown,hatbadge=gray,badge=gray,minimum size=1.5cm] (D) at (10,0.1) {${V^{\prime}_{2}}$};
\node[judge,mirrored,female,hair=yellow,minimum size=1.5cm] (B) at (2.5,0.1) {$J_{1}$};
\node[judge,hair=gray,minimum size=1.5cm] (C) at (7.5,0.1) {$J_{2}$};
%
\draw[yellow, ultra thick, ->] (1.91, 0.5) -- (0.57,0.5);
\draw[black, thick, ->] (1.9,0.5) -- (0.6,0.5);
\filldraw[black] (1.3,0.5) circle (0pt) node[anchor=south] {$w_{1}$};
\draw[black,fill=black!10, thick] (0.8,0) -- (0.8,-0.7) -- (1.6,-0.7) -- (1.6,0.1) -- (0.9,0.1) -- (A.mouth) -- cycle;
\draw[orange, <-] (1.5,0) -- (0.9,0);
\draw[red, ->] (1.5,-0.2) -- (0.9,-0.2);
\draw[orange, <-] (1.5,-0.4) -- (0.9,-0.4);
\draw[red, ->] (1.5,-0.6) -- (0.9,-0.6);
\draw[yellow, ultra thick, ->] (8.09, 0.5) -- (9.43,0.5);
\draw[black, thick, ->] (8.1,0.5) -- (9.4,0.5);
\filldraw[black] (8.7,0.5) circle (0pt) node[anchor=south] {$w_{2}$};
\draw[black,fill=black!10, thick] (9.2,0) -- (9.2,-0.7) -- (8.4,-0.7) -- (8.4,0.1) -- (9.1,0.1) -- (D.mouth) -- cycle;
\draw[gray, <-] (8.5,0) -- (9.1,0);
\draw[blue, ->] (8.5,-0.2) -- (9.1,-0.2);
\draw[gray, <-] (8.5,-0.4) -- (9.1,-0.4);
\draw[blue, ->] (8.5,-0.6) -- (9.1,-0.6);
\draw[black, thick, ->] (A.north) .. controls (1,1.7)  and (3,1.7)  .. (V.west);
\draw[black, thick, ->] (D.north) .. controls (9,1.7)  and (7,1.7)  .. (V.east);
\end{tikzpicture}
}
\end{center}

\begin{center}
\fbox{
\begin{tikzpicture}
\node[draw,rectangle,color=white,fill=white,minimum size=1.5cm] (U) at (5,-0.7) {$|{\mathbb {LOC}}\rangle$} ;
\node[draw,rectangle,color=black,fill=gray!40,minimum size=1.5cm] (V) at (5,1) { $\widehat{{\mathbb {S}}}\bigcup \widehat{{\mathbb {V}}^{\prime}}$ };
\node[builder,evil,shirt=black,hat=gray,hair=yellow,hatbadge=gray,badge=gray,minimum size=1.5cm] (A) at (0,0.1) {${S}_{1}$};
\node[builder,evil,mirrored,shirt=black,hat=gray,skin=black,hair=brown,hatbadge=gray,badge=gray,minimum size=1.5cm] (D) at (10,0.1) {${S}_{2}$};
\node[judge,mirrored,female,hair=yellow,minimum size=1.5cm] (B) at (2.5,0.1) {$J_{1}$};
\node[judge,hair=gray,minimum size=1.5cm] (C) at (7.5,0.1) {$J_{2}$};
%
\draw[yellow, ultra thick, ->] (1.91, 0.5) -- (0.57,0.5);
\draw[black, thick, ->] (1.9,0.5) -- (0.6,0.5);
\filldraw[black] (1.3,0.5) circle (0pt) node[anchor=south] {$w_{1}$};
\draw[black,fill=black!10, thick] (0.8,0) -- (0.8,-0.7) -- (1.6,-0.7) -- (1.6,0.1) -- (0.9,0.1) -- (A.mouth) -- cycle;
\draw[orange, <-] (1.5,0) -- (0.9,0);
\draw[red, ->] (1.5,-0.2) -- (0.9,-0.2);
\draw[orange, <-] (1.5,-0.4) -- (0.9,-0.4);
\draw[red, ->] (1.5,-0.6) -- (0.9,-0.6);
\draw[yellow, ultra thick, ->] (8.09, 0.5) -- (9.43,0.5);
\draw[black, thick, ->] (8.1,0.5) -- (9.4,0.5);
\filldraw[black] (8.7,0.5) circle (0pt) node[anchor=south] {$w_{2}$};
\draw[black,fill=black!10, thick] (9.2,0) -- (9.2,-0.7) -- (8.4,-0.7) -- (8.4,0.1) -- (9.1,0.1) -- (D.mouth) -- cycle;
\draw[gray, <-] (8.5,0) -- (9.1,0);
\draw[blue, ->] (8.5,-0.2) -- (9.1,-0.2);
\draw[gray, <-] (8.5,-0.4) -- (9.1,-0.4);
\draw[blue, ->] (8.5,-0.6) -- (9.1,-0.6);
\draw[black, thick, ->] (A.north) .. controls (1,1.7)  and (3,1.7)  .. (V.west);
\draw[black, thick, ->] (D.north) .. controls (9,1.7)  and (7,1.7)  .. (V.east);
\end{tikzpicture}
}
\end{center}

\begin{center}
\fbox{
\begin{tikzpicture}
\node[draw,rectangle,color=white,minimum size=1.5cm] at (5,-0.7) {};
\node[judge,female,hair=yellow,minimum size=1.5cm] (A) at (0.5,0) {$J_{1}$};
\node[judge,hair=gray,minimum size=1.5cm] (C) at (2.5,0) {$J_{2}$};
\node[judge,saturated,mirrored,minimum size=1.5cm] (B) at (5,0) {$J_{0}$};
\draw[yellow, ultra thick, ->] (4.41, 0.5) -- (0.97,0.5);
\draw[black, thick, ->] (4.4,0.5) -- (1,0.5);
\draw[yellow, ultra thick, ->] (4.41, 0.5) -- (2.97,0.5);
\draw[black, thick, ->] (4.4,0.5) -- (3.0,0.5);
\filldraw[black] (3.7,0.5) circle (0pt) node[anchor=south] {$w_{1},w_{2}$};
\draw[black,fill=black!10, thick] (1.1,0) -- (1.1,-0.7) -- (1.9,-0.7) -- (1.9,0.1) -- (1.2,0.1) -- (A.mouth) -- cycle;
\draw[orange, <-] (1.8,0) -- (1.2,0);
\draw[red, ->] (1.8,-0.2) -- (1.2,-0.2);
\draw[orange, <-] (1.8,-0.4) -- (1.2,-0.4);
\draw[red, ->] (1.8,-0.6) -- (1.2,-0.6);
%
\draw[black,fill=black!10, thick] (3.1,0) -- (3.1,-0.7) -- (3.9,-0.7) -- (3.9,0.1) -- (3.2,0.1) -- (C.mouth) -- cycle;
\draw[gray, ->] (3.8,0) -- (3.2,0);
\draw[blue, <-] (3.8,-0.2) -- (3.2,-0.2);
\draw[gray, ->] (3.8,-0.4) -- (3.2,-0.4);
\draw[blue, <-] (3.8,-0.6) -- (3.2,-0.6);
\end{tikzpicture}
}
\end{center}
\caption{\label{inter2} Interrogation or Simulation phase (top) followed by Distinction phase (bottom).}
\end{figure}
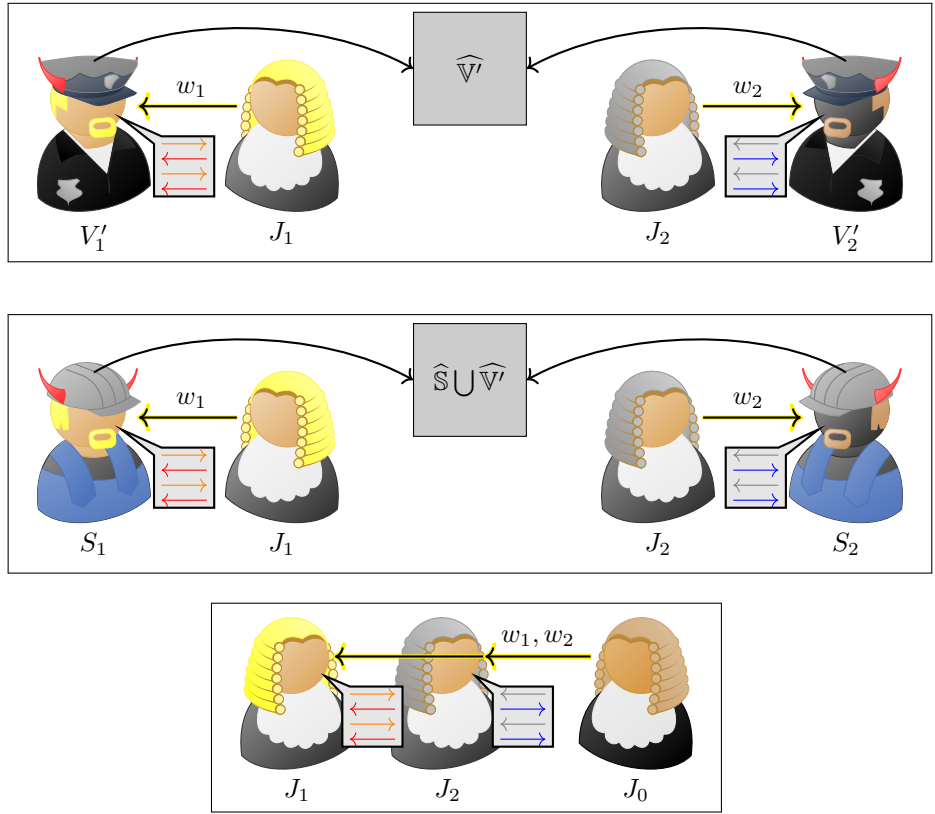

The results of this paper, depending on the specific bit commitment used, may be achieved under a stronger flavour of zero-knowledge
if a member of the non-locality class $\widehat{\mathbb {S}}$ is enough to break the binding property of the commitments. For instance, the result of section \ref{SEC:NEXP}
is really $\mathbf{ZK}^{\overset{\text{\tiny{poly}}}{\mathbb {NOSIG}}}\!\mathbf{MIP}^{{\mathbb {LOC}}}_{\underset{\text{\tiny{poly}}}{\mathbb {LOC}}} = \mathbf{NEXP}$ although existing proofs usually mean
$\mathbf{ZK}^{\overset{\text{\tiny{poly}}}{\mathbb {SIG}}}\!\mathbf{MIP}^{{\mathbb {LOC}}}_{\underset{\text{\tiny{poly}}}{\mathbb {LOC}}} = \mathbf{NEXP}$.
Using the bit commitment scheme based on the magic square game of \cite{Crepeau2017} we can also obtain $\mathbf{ZK}^{\overset{\text{\tiny{poly}}}{\mathbb {|LOC\rangle}}}\!\mathbf{MIP}^{{\mathbb {LOC}}}_{\underset{\text{\tiny{poly}}}{\mathbb {LOC}}} = \mathbf{NEXP}$.

Some interesting questions resulting from this definition is whether any higher class such as
$\mathbf{ZK}^{\overset{\text{\tiny{poly}}}{\mathbb {LOC}}}\!\mathbf{MIP}^{{\mathbb {LOC}}}_{\underset{\text{\tiny{poly}}}{\mathbb {LOC}}}$ or 
$\mathbf{ZK}^{\overset{\text{\tiny{poly}}}{\mathbb {NOSIG}}}\!\mathbf{MIP}^{{\mathbb {NOSIG}}}_{\underset{\text{\tiny{poly}}}{\mathbb {NOSIG}}}$ contains more than the natural examples
such as GRAPH ISO or CODE EQUIV already found in the most natural class 
$\mathbf{ZK}^{\overset{\text{\tiny{poly}}}{\mathbb {SIG}}}\!\mathbf{MIP}^{{\mathbb {SIG}}}_{\underset{\text{\tiny{poly}}}{\mathbb {SIG}}} = \mathbf{ZK}\!\mathbf{IP}$.

\subsection{A note on notation}
$$\mathlarger{\mathlarger{\mathlarger{\mathbf{ZK}^{{\mathbb {S}}}\!\!\;\mathbf{MIP}^{{\mathbb {P}}}_{{\mathbb {V}}}}}}$$ is the complexity class of Zero-Knowledge Multi-provers Interactive Proofs where (honest and dishonest) provers are restricted to non-locality class $\mathbb {P}$ (important for soundness),
where the honest verifier is from non-locality class $\mathbb {V}$ (also important for soundness), and where the Zero-Knowledge simulators are from non-locality class $\mathbb {S}$ unless $\widehat{V^\prime}$ is outside of $\mathbb {S}$ in which case they are from the class of $\widehat{V^\prime}$.

\newcommand{\nneq}[0]{\neq\!\!\!\!\neq}
\newcommand{\inn}[0]{\!\in\!}

\section{CMOSSY 3-COL Honest-Verifer Zero-Knowledge Interactive Proof}\label{SEC:appendix_A}

\begin{quote}
\rule{\linewidth}{1pt}
\begin{protocol}\label{MIP_3COL23A} Two-out-of-Three-Prover, 3-COL.\end{protocol}
The verifiers $V_{1},V_{2},V_{3}$ pre-agree on random edges $(n_{0},n_{1})$ and $(n_{2},n_{3})$, random strings $r_{0},r_{1},r_{2},r_{3}\neq 0$ and 
the provers $P_{1},P_{2},P_{3}$ pre-agree on random values $b_{n_{i}} : n_{i} \inn V$ and a random 3-colouring of $G$:
$\left\{ c_{n_{i}} \inn \{0,1,2 \} : n_{i} \inn V \right\}$ such that $(n_{i},n_{j}) \inn E \implies c_{n_{j}} \neq c_{n_{i}}$.
They also pre-compute an array $W[n_{i},r] := b_{n_{i}} \cdot r + c_{n_{i}} : n_{i} \inn V, r\inn \{1,2\}$.
The values $(n_{0},n_{1},r_{0},r_{1}),(n_{2},n_{3},r_{2},r_{3})$ are selected under one of three constraints: either 
$$(n_{0}, n_{1}) = (n_{2}, n_{3}), r_{0}\neq r_{2}, r_{1}\neq r_{3} \text{ or }$$
$$\exists i,j \inn \{0,1\}\times \{2,3\} : n_{i}=n_{j},r_{i}=r_{j} \text{ or }$$
$$(n_{0}, n_{1}) = (n_{2}, n_{3}), (r_{0}, r_{1}) = (r_{2}, r_{3}).$$

The verifiers $V_{1},V_{2},V_{3}$ pre-select $P_{A},P_{B}$ at random from $P_{1},P_{2},P_{3}$.

{\bf Commit phase:}

\begin{itemize}

\item $P_{A}$ receives nodes $n_{0},n_{1}$, strings $r_{0},r_{1}$ from $V_{A}$ and if $(n_{0},n_{1}) \inn E$, replies $W[n_{0},r_{0}],W[n_{1},r_{1}]$.

\item $P_{B}$ receives nodes $n_{2},n_{3}$, strings $r_{2},r_{3}$ from $V_{B}$ and if $(n_{2},n_{3}) \inn E$, replies $W[n_{2},r_{2}],W[n_{3},r_{3}]$.

\end{itemize}

{\bf Check phase:}

\begin{itemize}

\item[] {\bf Consistency Test:}
\item if $(n_{0}, n_{1}) = (n_{2}, n_{3}), (r_{0}, r_{1}) = (r_{2}, r_{3})$ then $V_A, V_B$ accept iff $$(W[n_{0},r_{0}],W[n_{1},r_{1}]) = (W[n_{2},r_{2}], W[n_{3},r_{3}]).$$

\item[] {\bf Edge-Verification Test:}
\item  if $(n_{0}, n_{1}) = (n_{2}, n_{3}), r_{0}\neq r_{2}, r_{1}\neq r_{3}$ then $V_A, V_B$ accept iff $$W[n_{0},r_{0}]+W[n_{2},r_{2}] \neq W[n_{1},r_{1}]+W[n_{3},r_{3}].$$

\item[] {\bf Well-Definition Test:}
\item  if $\exists i,j \inn \{0,1\}\times \{2,3\} : n_{i}=n_{j},r_{i}=r_{j}$ then $V_A, V_B$ accept iff $W[n_{i},r_{i}] = W[n_{j},r_{j}]$.

\end{itemize}

\rule{\linewidth}{1pt}
\end{quote}

\begin{quote}
\rule{\linewidth}{1pt}
\begin{protocol}\label{SIM_3COL3A} HV Two-prover simulation.\end{protocol}
{\bf Commit phase:}

\begin{itemize}

\item Let $\pi$ be a uniform permutation of $\{0,1,2\}$ and let $coco := 0$.
\item $\forall n \inn V, r \inn \{0,1,2\}$, let $mark[n,r] := false$, $count[n] := 0$, $colour[r]:= \pi(r)$.

\item $S$ runs $V_{1},V_{2},V_{3}$ until it receives $(n_{2A-2},n_{2A-1},r_{2A-2},r_{2A-1})$,\\ $(n_{2B-2},n_{2B-1},r_{2B-2},r_{2B-1})$ from $V_{A},V_{B}$.

\item Whenever $(n_{2i-2},n_{2i-1}) \inn E$ is provided by $V_{i}$,\\
$S$ replies $(w_{2i-2},w_{2i-1})$,
both computed as follows for $k\inn \{2i-2,2i-1\}$:
\begin{itemize}
\item If $\neg mark[n_{k},r_{k}]$ then
\begin{itemize}
\item If $count[n_{k}] = 0$ then pick $W[n_{k},r_{k}]$ uniformly in $\{0,1,2\}$.
\item If $count[n_{k}] = 1$ then
\begin{itemize}
\item Let $W[n_{k},r_{k}] := - colour[coco] - W[n_{k},-r_{k}]$
\item Let $coco := coco+1$.
\end{itemize}

\item  Let $mark[n_{k},r_{k}] := true$, $count[n_{k}] := count[n_{k}]+1$.
\end{itemize}
\item Let $w_{k} := W[n_{k},r_{k}]$.
\end{itemize}

\end{itemize}

\rule{\linewidth}{1pt}
\end{quote}


\end{document}